\pdfoutput=1

\newif\ifdraft\draftfalse 
\newif\ifanon\anonfalse   
\newif\iffull\fullfalse    
\newif\iflongrefs\longrefsfalse 
\newif\ifbackref\backreffalse 
\newif\ifsooner\soonerfalse
\newif\iflater\laterfalse
\newif\ifcamera\cameratrue     
\newif\ifcheckpagebudget\checkpagebudgetfalse
\newif\ifconference\conferencetrue
\newif\ifpages\pagesfalse      

\makeatletter \@input{texdirectives.tex} \makeatother

\documentclass[runningheads]{llncs}

\usepackage{afterpage}
\usepackage{color}
\usepackage{soul}
\usepackage{amssymb,amsmath,amsfonts}
\usepackage{extarrows}
\usepackage{version}
\usepackage{xspace}
\usepackage{graphicx}
\usepackage{listings}
\usepackage{wrapfig}
\usepackage{ifpdf}
\usepackage{semantic}
\usepackage{mathpartir} 

\usepackage[numbers,sort]{natbib}
\newcommand\citepos[1]{\citeauthor{#1}'s\ \citeyear{#1}}
\usepackage{stmaryrd}
\usepackage{subfigure}
\usepackage{multirow}
\usepackage{proof}
\usepackage{framed}
\usepackage{caption}
\usepackage{mathpartir}
\usepackage{enumitem}
\usepackage{tikz}
\usetikzlibrary{positioning}
\usetikzlibrary{arrows}
\usepackage[T1]{fontenc}
\usepackage[utf8]{inputenc}

\newcommand{\kw}[1]{\text{\ls|#1|}}
\newcommand\Goal[3]{\mbox{\ensuremath{#1 \vdash \kw{#2} : \kw{#3}}\xspace}}
\newcommand\VGoal[2]{\mbox{\ensuremath{#1 \vDash \kw{#2}}\xspace}}

\definecolor{dkblue}{rgb}{0,0.1,0.5}
\definecolor{dkgreen}{rgb}{0,0.4,0}
\definecolor{dkred}{rgb}{0.6,0,0}
\definecolor{dkpurple}{rgb}{0.7,0,1.0}
\definecolor{purple}{rgb}{0.9,0,1.0}
\definecolor{olive}{rgb}{0.4, 0.4, 0.0}
\definecolor{teal}{rgb}{0.0,0.4,0.4}
\definecolor{azure}{rgb}{0.0, 0.5, 1.0}
\definecolor{gray}{rgb}{0.5, 0.5, 0.5}
\definecolor{dkgrey}{rgb}{0.2, 0.2, 0.2}
\definecolor{lilac}{rgb}{0.70, 0.04, 0.08}
\definecolor{dkyellow}{rgb}{0.5,0.5,0}

\usepackage{hyperref}

\hypersetup{breaklinks=true}
\hypersetup{colorlinks=true}
\hypersetup{linkcolor=black}
\hypersetup{citecolor=black}
\hypersetup{filecolor=black}
\hypersetup{urlcolor=dkblue}
\newcommand\maybecolor[1]{\color{#1}}

\let\ls\lstinline

\def\Snospace~{\S{}}

\newcommand\fstar{F$^\star$\xspace}
\newcommand\emf{EMF$^\star$\xspace}
\newcommand\metafstar{Meta-F$^\star$\xspace}
\newcommand\lowstar{Low$^\star$\xspace}

\newcommand{\comm}[3]{\ifcheckpagebudget\else\ifdraft{{\color{#1}[#2: #3]}}\fi\fi}
\newcommand{\nik}[1]{\comm{dkpurple}{Nik}{#1}}
\newcommand{\ch}[1]{\comm{teal}{CH}{#1}}
\newcommand{\aseem}[1]{\comm{magenta}{Aseem}{#1}}
\newcommand{\guido}[1]{\comm{blue}{Guido}{#1}}

\newcommand{\tahina}[1]{\comm{dkred}{Tahina}{#1}}

\newcommand{\cpc}[1]{\comm{brown}{Cl\'ement}{#1}}

\begin{document}
\def\titlestr{\metafstar: Proof Automation with \\SMT, Tactics, and Metaprograms}
\title{\titlestr}
%

\ifpages
\let\oldsection\section
\renewcommand\section{\clearpage\oldsection}
\fi

\author{
    Guido Mart\'inez\inst{1,2} \and
    Danel Ahman\inst{3} \and
    Victor Dumitrescu\inst{4} \and
    Nick Giannarakis\inst{5} \and
    Chris Hawblitzel\inst{6} \and
    C\u{a}t\u{a}lin Hri\c{t}cu\inst{2} \and
    Monal Narasimhamurthy\inst{7} \and\\
    Zoe Paraskevopoulou\inst{5} \and
    Cl\'ement Pit-Claudel\inst{8} \and
    Jonathan Protzenko\inst{6} \and\\
    Tahina Ramananandro\inst{6} \and
    Aseem Rastogi\inst{6} \and
    Nikhil Swamy\inst{6}
}
\authorrunning{Mart\'inez et al.}
\institute{
    $^1$CIFASIS-CONICET \quad
    $^2$Inria Paris \quad
    $^3$University of Ljubljana \\
    $^4$MSR-Inria Joint Centre \quad
    $^5$Princeton University \quad
    $^6$Microsoft Research \quad
    $^7$University of Colorado Boulder \quad
    $^8$MIT CSAIL
}
\maketitle              
\begin{abstract}

%

    We introduce \metafstar, a tactics and metaprogramming
    framework for the \fstar program verifier.
    The main novelty of \metafstar is allowing the use of tactics and 
    metaprogramming to discharge assertions not
    solvable by SMT, or to just simplify them into well-behaved SMT
    fragments. Plus, \metafstar can be used to
    generate verified code automatically.

\smallskip
    \metafstar is implemented as an \fstar \emph{effect}, which, given
    the powerful effect system of \fstar{}, heavily increases code reuse
    and even enables the lightweight verification of metaprograms.
    Metaprograms can be either interpreted, or compiled to efficient
    native code that can be dynamically loaded into the \fstar type-checker
    and can interoperate with interpreted code.
    Evaluation on realistic case studies shows that \metafstar
    provides substantial gains in proof development, efficiency, and robustness.

\ifcamera
\fi

\end{abstract}

\section{Introduction}


Scripting proofs using tactics and metaprogramming has a long
tradition in interactive theorem provers (ITPs), starting with Milner's
Edinburgh LCF~\cite{gordon-edinburghLCF:1979}.  In this lineage,
properties of \emph{pure} programs are specified
in expressive higher-order (and often dependently typed) logics,
and proofs are conducted using various imperative programming
languages, starting originally with ML.

Along a different axis, program verifiers like
Dafny~\cite{leino10dafny}, VCC~\cite{CohenMST10},
Why3~\cite{filliatre13why3}, and 
Liquid Haskell~\cite{liquidhaskell} 
target both pure {\em and effectful} programs,
with side-effects ranging from divergence to
concurrency, but provide relatively weak logics for
specification (e.g., first-order
logic with a few selected theories like linear arithmetic).
They work primarily by computing verification conditions (VCs)
from programs, usually relying on annotations such as pre- and postconditions,
and encoding them to automated theorem provers (ATPs) such as
satisfiability modulo theories (SMT) solvers,
often providing excellent automation.

These two sub-fields have influenced one another, though
the situation is somewhat asymmetric.
On the one hand, most interactive provers have gained support for
exploiting SMT solvers or other ATPs, providing push-button automation
for certain kinds of assertions~\cite{CzajkaK17, smtcoqCAV2017,
PaulsonB10, KaliszykU14, KaliszykU15a}.
%
%
On the other hand, recognizing the importance of interactive proofs,
Why3~\cite{filliatre13why3} interfaces with ITPs like Coq.
However, working over proof obligations translated from Why3
requires users to be familiar not only with both these systems, but
also with the specifics of the translation.
%
%
And beyond Why3 and the tools based on it~\cite{CuoqKKPSY12}, no other
SMT-based program verifiers have full-fledged support for interactive
proving, leading to several downsides:

\paragraph{\bf Limits to expressiveness} The expressiveness of program
verifiers can be limited by the ATP used.
When dealing with theories that are undecidable and
difficult to automate (e.g., non-linear arithmetic or
separation logic), proofs in ATP-based systems 
may become impossible or, at best, extremely tedious.
\guido{citations would be nice}

\paragraph{\bf Boilerplate} To work around this lack of automation,
programmers have to construct detailed proofs by hand, often repeating
many tedious yet error-prone steps, so as to provide hints to the
underlying solver to discover the proof. In contrast, ITPs
with metaprogramming facilities excel at expressing
domain-specific automation to complete such tedious proofs.

\paragraph{\bf Implicit proof context} In most program verifiers,
the logical context of a proof
is implicit in the program text and depends on the control flow and
the pre- and postconditions of preceding computations. Unlike in
interactive proof assistants, programmers have no explicit access,
neither visual nor programmatic, to this context, making proof
structuring and exploration extremely difficult.
\\



In direct response to these drawbacks, we seek a system that successfully 
combines the convenience of an
automated program verifier for the common case, while seamlessly
transitioning to an interactive proving experience for those parts of
a proof that are hard to automate. Towards this end, we
propose \metafstar, a tactics and metaprogramming framework for
the \fstar~\cite{mumon,dm4free} program verifier.

\subsection*{Highlights and Contributions of \metafstar}

\fstar has historically been more deeply rooted as an
SMT-based program verifier.
Until now, \fstar discharged VCs
exclusively by calling an SMT solver (usually Z3~\cite{MouraB08}), 
providing good automation for many common
program verification tasks, but
also exhibiting the drawbacks discussed above.


\metafstar is a framework that allows \fstar users to manipulate VCs
using \emph{tactics}. More generally, it
supports \emph{metaprogramming}, allowing programmers to script
the construction of programs, by manipulating their
syntax and customizing the way they are type-checked. This allows
programmers to (1) implement custom procedures for manipulating VCs;
(2) eliminate boilerplate in proofs and programs; and (3) to inspect
the proof state visually and to manipulate it programmatically,
addressing the drawbacks discussed above.
SMT still plays a central role in \metafstar: a typical usage involves
implementing tactics to transform VCs, so as to bring them into theories
well-supported by SMT, without needing to (re)implement full
decision procedures.
Further, the generality of \metafstar allows implementing non-trivial
language extensions (e.g., typeclass resolution) entirely as
metaprogramming libraries, without changes to the \fstar{}
type-checker.

The technical {\bf contributions} of our work include the following:



%

\paragraph*{\bf ``Meta-'' is just an effect (\autoref{sec:tac-effect})}
%
%
\metafstar is implemented using \fstar's extensible effect
system, which keeps programs and
metaprograms properly isolated.
%
%
Being first-class \fstar programs, metaprograms
are typed, call-by-value, direct-style, higher-order functional
programs, much like the original ML.
%
Further, metaprograms can be themselves verified (to a degree, see
\autoref{sec:specifying-metaprograms}) and metaprogrammed.

\paragraph*{\bf Reconciling tactics with VC generation (\autoref{sec:vcgen})}
In program verifiers the programmer often guides
the solver towards the proof by supplying intermediate 
assertions. \metafstar retains this style,
but additionally allows assertions to be solved by 
tactics. To this end, a contribution of
our work is extracting, from a VC, a proof state
encompassing all relevant hypotheses, including those implicit
in the program text.
%

\paragraph*{\bf Executing metaprograms efficiently (\autoref{sec:evaluation})}
Metaprograms are executed during type-checking.
As a baseline, they can be interpreted using \fstar's existing
(but slow) abstract machine for term normalization, or
a faster normalizer based on normalization by evaluation
(NbE)~\cite{bergerNbe, BoespflugDG11}.
%
%
For much faster execution speed, metaprograms can also be run natively.
This is achieved by combining the existing extraction mechanism of
\fstar to OCaml with a new framework for safely extending the \fstar
type-checker with such native code.
%
%

\paragraph*{\bf Examples (\autoref{sec:by-example}) and evaluation (\autoref{sec:experiments})}
We evaluate \metafstar on several case studies.
First, we present a functional correctness proof for the Poly1305
message authentication code (MAC)~\cite{poly1305}, using a novel combination
of proofs by reflection for dealing with non-linear arithmetic and SMT
solving for linear arithmetic. We measure a clear gain in proof
robustness: SMT-only proofs succeed only rarely
(for reasonable timeouts),
whereas our tactic+SMT proof is concise, never fails,
and is faster.
%
Next, we demonstrate an improvement in expressiveness,
by developing a small library for proofs of heap-manipulating programs in
separation logic, which was previously out-of-scope for \fstar.
Finally, we illustrate the ability to automatically construct verified
effectful programs, by introducing a library for metaprogramming verified
low-level parsers and serializers with applications to network
programming, where verification is accelerated by processing
the VC with tactics,
and by programmatically tweaking the SMT context.

\medskip
We conclude that tactics and metaprogramming
can be prosperously combined with VC generation and SMT solving
to build verified programs
with better, more scalable, and more robust automation.

\medskip
The full version of this paper, including appendices, can be found online in
{\url{https://www.fstar-lang.org/papers/metafstar}}.

\section{\metafstar by Example}
\label{sec:by-example}

\fstar is a general-purpose
programming language aimed at program verification.
It puts together the automation of an SMT-backed deductive
verification tool with the expressive power of a language
with full-spectrum dependent types.
Briefly, it is a functional, higher-order, effectful, dependently typed
language, with syntax loosely based on OCaml.
\fstar supports refinement types and Hoare-style specifications,
computing VCs of computations via a type-level weakest
precondition (WP) calculus packed within \emph{Dijkstra
monads}~\cite{fstar-pldi13}.
\fstar's effect system is also user-extensible~\cite{dm4free}.
Using it, one can model or embed 
imperative programming in
styles ranging from ML to
C~\cite{lowstar} and
assembly~\cite{valefstar}.
After verification, \fstar programs can be extracted to efficient
OCaml or F\# code. A first-order fragment of \fstar, 
called \lowstar, can also be extracted to C via the KreMLin compiler~\cite{lowstar}.

This paper introduces \metafstar, a metaprogramming framework for
\fstar that allows users to safely customize and extend \fstar in many
ways.
%
For instance, \metafstar can be used to preprocess or solve proof
obligations; synthesize \fstar expressions; generate
top-level definitions; and resolve implicit arguments in user-defined
ways, enabling non-trivial extensions.
%
This paper primarily discusses the first two features.
Technically, none of these features deeply increase the expressive power
of \fstar, since one could manually program in \fstar terms that can
now be metaprogrammed. However, as we will see shortly, manually
programming terms and their proofs can be so prohibitively costly as
to be practically infeasible.


\metafstar is similar to other tactic frameworks, such as
Coq's~\cite{Delahaye00} or Lean's~\cite{EbnerURAM17}, in presenting a
set of goals to the programmer, providing commands to break them down,
allowing to inspect and build abstract syntax, etc.
In this paper,
we mostly detail the characteristics where \metafstar \emph{differs}
from other engines.


This section presents \metafstar informally, displaying its usage
through case studies. We present any necessary \fstar background as needed.

\subsection{Tactics for Individual Assertions and Partial Canonicalization}
\label{sec:non-linear}



Non-linear arithmetic reasoning is crucially needed for the
verification of optimized, low-level cryptographic
primitives~\cite{vale,haclstar}, an important use case
for \fstar~\cite{everest} and other verification frameworks, including
those that rely on SMT solving alone (e.g., Dafny~\cite{leino10dafny})
as well as those that rely exclusively on tactic-based proofs (e.g.,
FiatCrypto~\cite{fiat-crypto}).
While both styles have demonstrated significant successes, we make
a case for a middle ground, leveraging the SMT solver for
the parts of a VC where it is effective,
and using tactics only where it is not.

We focus on Poly1305~\cite{poly1305}, a widely-used cryptographic MAC
that computes a series of integer multiplications and additions modulo
a large prime number
$p = 2^{130} - 5$.  Implementations of the Poly1305 multiplication and
mod operations are carefully hand-optimized to represent 130-bit
numbers in terms of smaller 32-bit or 64-bit registers, using clever
tricks;
proving their correctness
requires reasoning about long sequences of
additions and multiplications.

\paragraph{\bf Previously: Guiding SMT solvers by manually
  applying lemmas}

Prior proofs of correctness of Poly1305 and other cryptographic
primitives using SMT-based program verifiers, including
\fstar~\cite{haclstar} and Dafny~\cite{vale},
use a combination of SMT automation and manual application of lemmas.
On the plus side, SMT solvers are excellent at linear arithmetic, so
these proofs delegate all associativity-commutativity (AC) reasoning
about addition to SMT.
Non-linear arithmetic in SMT solvers, even just AC-rewriting and
distributivity, are, however, inefficient and unreliable---so much so
that the prior efforts above (and other works too~\cite{Ironclad,
  IronFleet})
simply turn off support for non-linear arithmetic
in the solver, in order not to degrade verification
performance across the board due to poor interaction of theories.
%
Instead, users need to explicitly invoke lemmas.\footnote{\ls$Lemma (requires
pre) (ensures post)$ is \fstar notation for the type of a 
computation proving \ls$pre ==> post$---we omit \ls$pre$ when it is
trivial. In \fstar's standard library, math lemmas are proved using
SMT with little or no interactions between problematic theory
combinations. These lemmas can then be explicitly invoked in larger
contexts, and are deleted during extraction.}

For instance, here is a statement and proof of a lemma about Poly1305
in \fstar. The property and its proof do not really
matter; the lines marked ``\ls$(* argh! *)$'' do. In this particular proof,
working around the solver's inability to effectively reason about
non-linear arithmetic, the programmer has spelled out basic facts about
distributivity of multiplication and addition, by calling the library
lemma \ls$distributivity_add_right$, in order to guide
the solver towards the proof.
(Below, \ls$p44$ and \ls$p88$ represent $2^{44}$ and $2^{88}$ respectively)
\begin{lstlisting}
let lemma_carry_limb_unrolled (a0 a1 a2 : nat) : Lemma (ensures (
   a0 % p44 + p44 * ((a1 + a0 / p44) % p44) + p88 * (a2 + ((a1 + a0 / p44) / p44))
   == a0 + p44 * a1 + p88 * a2)) =
 let z = a0 % p44 + p44 * ((a1 + a0 / p44) % p44)
           + p88 * (a2 + ((a1 + a0 / p44) / p44)) in
 distributivity_add_right p88 a2 ((a1 + a0 / p44) / p44); (* argh! *)
 pow2_plus 44 44;
 lemma_div_mod (a1 + a0 / p44) p44;
 distributivity_add_right p44 ((a1 + a0 / p44) % p44)
           (p44 * ((a1 + a0 / p44) / p44)); (* argh! *)
 assert (p44 * ((a1 + a0 / p44) % p44) + p88 * ((a1 + a0 / p44) / p44)
           == p44 * (a1 + a0 / p44) );
 distributivity_add_right p44 a1 (a0 / p44); (* argh! *)
 lemma_div_mod a0 p44
\end{lstlisting}
Even at this relatively small scale, needing to explicitly instantiate
the distributivity lemma is verbose and error prone.
Even worse, the user is blind while doing so: the
program text does not display the current set of available facts nor
the final goal.
Proofs at this level of abstraction are painfully detailed in some
aspects, yet also heavily reliant on the SMT solver to fill in the aspects of
the proof that are missing.
%
%

Given enough time, the solver can sometimes find a proof without 
the additional hints, but this is usually rare and dependent on context,
and almost never robust.
In this particular example we find by varying Z3's random
seed that, in an isolated setting,
the lemma is proven automatically about 32\% of the time.
%
%
The numbers are much worse for more complex proofs, and where the
context contains many facts, making this style quickly spiral out of
control.
For example, a proof of one of the main lemmas in
Poly1305, \ls$poly_multiply$, requires 41 steps of rewriting for
associativity-commutativity of multiplication, and
distributivity of addition and multiplication---making the proof
much too long to show here.
%


\paragraph{\bf SMT and tactics in \metafstar}

\label{sec:canon_semiring}
\label{sec:poly_multiply}

The listing below shows the statement and proof of \ls$poly_multiply$
in \metafstar, of which the lemma above was previously only a small part.
Again, the specific property proven is not particularly
relevant to our discussion. But, this time, the proof contains just
two steps.
\begin{lstlisting}
let poly_multiply (n p r h r0 r1 h0 h1 h2 s1 d0 d1 d2 h1 h2 hh : int) : Lemma
  (requires p > 0 /\ r1 >= 0 /\ n > 0 /\ 4 * (n * n) == p + 5 /\ r == r1 * n + r0 /\
            h == h2 * (n * n) + h1 * n + h0 /\ s1 == r1 + (r1 / 4) /\ r1 % 4 == 0 /\
            d0 == h0 * r0 + h1 * s1 /\ d1 == h0 * r1 + h1 * r0 + h2 * s1 /\
            d2 == h2 * r0 /\ hh == d2 * (n * n) + d1 * n + d0)
  (ensures (h * r) % p == hh % p) =
  let r14 = r1 / 4 in
  let h_r_expand = (h2 * (n * n) + h1 * n + h0) * ((r14 * 4) * n + r0) in
  let hh_expand = (h2 * r0) * (n * n) + (h0 * (r14 * 4) + h1 * r0
                            + h2 * (5 * r14)) * n + (h0 * r0 + h1 * (5 * r14)) in
  let b = (h2 * n + h1) * r14 in
  modulo_addition_lemma hh_expand p b;
  assert (h_r_expand == hh_expand + b * (n * n * 4 + ($-$5)))
      by (canon_semiring int_csr) (* Proof of this step by Meta--F* tactic *)
\end{lstlisting}

First, we call a single lemma about modular addition from \fstar's
standard library.
Then, we assert an equality annotated with a tactic (\ls$assert..by$).
Instead of encoding the assertion as-is to the SMT solver,
it is preprocessed by the \ls$canon_semiring$ tactic.
The tactic is presented with the asserted equality as its
goal, in an environment containing not only all variables in scope
but also hypotheses for the precondition of \ls$poly_multiply$
and the postcondition of the \ls$modulo_addition_lemma$ call
(otherwise, the assertion could not be proven).
The tactic will then canonicalize the sides of the equality, but notably
only ``up to'' linear arithmetic conversions.
Rather than fully canonicalizing the terms, the tactic just rewrites them
into a sum-of-products canonical form, leaving all the remaining
work to the SMT solver,
which can then easily and robustly discharge the
goal using linear arithmetic only.

%
This tactic works over terms in the commutative semiring of integers
(\ls$int_csr$) using proof-by-reflection~\cite{GregoireM05, gonthier2008formal,ChaiebN08,Besson06}.
Internally, it is composed of a simpler, also proof-by-reflection based 
tactic \ls$canon_monoid$ that
works over monoids, which is then ``stacked'' on itself to build
\ls$canon_semiring$.
The basic idea of proof-by-reflection is to reduce most of the proof
burden to mechanical computation, obtaining much more efficient proofs
compared to repeatedly applying lemmas.
For \ls$canon_monoid$, we begin with
a type for monoids,
a small AST representing monoid values, and
a denotation for expressions back into the monoid type.
\begin{lstlisting}
type monoid (a:Type) = { unit : a; mult : (a -> a -> a); (* + monoid laws ... *) }
type exp (a:Type) = | Unit : exp a | Var : a -> exp a | Mult : exp a -> exp a -> exp a
(* Note on syntax: $\mathsf{\#a}$ below denotes that $\mathsf{a}$ is an implicit argument *)
let rec denote (#a:Type) (m:monoid a) (e:exp a) : a =
  match e with
  | Unit -> m.unit | Var x -> x | Mult x y -> m.mult (denote m x) (denote m y)
\end{lstlisting}
To canonicalize an \ls$exp$, it is first converted to a
list of operands (\ls$flatten$) and then reflected back to the monoid
(\ls$mldenote$).
The process is proven correct, in the particular case
of equalities, by the \ls$monoid_reflect$ lemma.
%
%
%
%
\begin{lstlisting}
val flatten : #a:Type -> exp a -> list a
val mldenote : #a:Type -> monoid a -> list a -> a
let monoid_reflect (#a:Type) (m:monoid a) (e_1 e_2 : exp a)
              : Lemma (requires (mldenote m (flatten e_1) == mldenote m (flatten e_2)))
                       (ensures (denote m e_1 == denote m e_2)) = ...
\end{lstlisting}
At this stage, if the goal is \ls$t_1 == t_2$, we require two
monoidal expressions \ls$e_1$ and \ls$e_2$ such that
\ls$t_1 == denote m e_1$ and \ls$t_2 == denote m e_2$.
They are constructed by the tactic \ls$canon_monoid$ by inspecting
the \emph{syntax} of the goal, using \metafstar's reflection
capabilities (detailed ahead in~\autoref{sec:quoting}).
We have no way to prove once and for all that the expressions built by
\ls$canon_monoid$ correctly denote the terms, but this fact can be
proven automatically at each application of the tactic, by simple unification.
The tactic then applies the lemma \ls$monoid_reflect m e_1 e_2$,
and the goal is changed to \ls$mldenote m (flatten e_1) ==$ \ls$mldenote m (flatten e_2)$.
Finally, by normalization, each side will be
canonicalized by running \ls$flatten$ and \ls$mldenote$.

%
%
The \ls$canon_semiring$ tactic follows a similar approach,
and is similar to existing reflective
tactics for other proof assistants~\cite{ring, GregoireM05},
except that it only canonicalizes up to linear arithmetic,
as explained above.
The full VC for \ls$poly_multiply$
contains many other facts, e.g.,
that \ls$p$ is non-zero so the division is well-defined
and that the postcondition does indeed hold.
These obligations remain in a ``skeleton'' VC that is also easily proven
by Z3.
This proof is much easier for the programmer to write and much more
robust, as detailed ahead in \autoref{sec:canonicalization}.
The proof of Poly1305's other main lemma, \ls$poly_reduce$, is also
similarly well automated.

\paragraph{\bf Tactic proofs without SMT}

Of course, one can verify
\ls$poly_multiply$ in Coq, following the same conceptual proof used in \metafstar,
but relying on tactics only.
Our proof
\iffull
(see Fig.~\ref{fig:coqpolymultiply} in the appendix)
\else
(included in the appendix)
\fi
is 27 lines long, two of which involve the use of Coq's \ls$ring$
tactic (similar to our \ls$canon_semiring$ tactic) and \ls$omega$ tactic
for solving formulas in Presburger arithmetic. 
The remaining 25 lines include steps to destruct the
propositional structure of terms, rewrite by equalities, enriching the
context to enable automatic modulo rewriting (Coq does not fully
automatically recognize equality modulo $p$ as an equivalence relation
compatible with arithmetic operators).
While a mature proof assistant like Coq has libraries and tools to
ease this kind of manipulation, it can still be
verbose.

In contrast, in \metafstar{} all of these mundane parts of a proof are
simply dispatched to the SMT solver, which decides
linear arithmetic efficiently, beyond the quantifier-free Presburger fragment
supported by tactics like \ls$omega$, handles congruence closure natively, etc.

\subsection{Tactics for Entire VCs and Separation Logic}
\label{sec:seplogic}

A different way to invoke \metafstar is over an entire VC.
%
While the exact shape of VCs is hard to predict, users with some
experience can write tactics that find and solve particular
sub-assertions within a VC, or simply massage them into shapes better
suited for the SMT solver. We illustrate the idea on proofs for
heap-manipulating programs.

\guido{We should be careful on a previous criticism:
are we doing SL or using the standard heap model?
I think we wanted to claim that is advantage to interop
with existing code}

One verification method that has eluded \fstar until now is
separation logic, the main reason being that the pervasive ``frame
rule'' requires instantiating existentially quantified heap variables,
which is a challenge for SMT solvers, and simply too tedious for users.
With \metafstar, one can do better.
We have written a (proof-of-concept) embedding of separation logic
and a tactic (\ls$sl_auto$) that performs heap frame inference automatically.

The approach we follow consists of designing the WP specifications for
primitive stateful actions so as to make their footprint
syntactically evident.
%
The tactic then descends through VCs
until it finds an existential for heaps arising from the frame rule.
Then, by solving an equality between heap expressions (which requires
canonicalization, for which we use a variant of \ls$canon_monoid$
targeting {\em commutative} monoids)
the tactic finds the frames and instantiates the existentials.
\ch{I
  expect that previous tactics for SL in Coq would also do this,
  although I also guess they have to do a lot more. Didn't Danel
  add some references at some point?}%
Notably, as opposed to other tactic frameworks for separation
logic~\cite{Appel06, McCreight09, KrebbersTB17, nmb08htt}, this is \emph{all}
our tactic does before dispatching to the SMT solver, which can now be
effective over the instantiated VC.
%

We now provide some detail on the framework.
Below,
`\ls$emp$' represents the empty heap,
`$\bullet$' is the separating conjunction and
`\ls$r |-> v$' is the heaplet with the single reference \ls$r$ set to value \ls$v$.%
\footnote{This differs from the usual presentation where these three operators
are heap \emph{predicates} instead of heaps.}
Our development distinguishes between a ``heap''
and its ``memory'' for technical reasons, but we will treat
the two as equivalent here.
%
%
Further, \ls$defined$ is a predicate discriminating valid heaps
(as in~\cite{NanevskiVB10}), i.e., 
those built from separating conjunctions of
\emph{actually} disjoint heaps.

We first define the type of WPs and present the WP for the
frame rule:
\begin{lstlisting}
let pre = memory -> prop (* predicate on initial heaps *)
let post a = a -> memory -> prop (* predicate on result values and final heaps *)
let wp a = post a -> pre (* transformer from postconditions to preconditions *)

let frame_post (#a:Type) (p:post a) (m_0:memory) : post a =
  fun x m_1 -> defined (m_0 <*> m_1) /\ p x (m_0 <*> m_1)
let frame_wp (#a:Type) (wp:wp a) (post:post a) (m:memory) =
  exists m_0 m_1. defined (m_0 <*> m_1) /\ m == (m_0 <*> m_1) /\ wp (frame_post post m_1) m_0
\end{lstlisting}
Intuitively, \ls$frame_post p m_0$ behaves as the postcondition \ls$p$
``framed'' by \ls$m_0$, i.e., 
\ls@frame_post p m$_0$ x m$_1$@ holds when the
two heaps \ls$m_0$ and \ls$m_1$ are disjoint and \ls$p$ holds over the result value \ls$x$ and the conjoined heaps.
Then, \ls$frame_wp wp$ takes a postcondition \ls$p$ and initial heap
\ls$m$, and requires that \ls$m$ can be split into disjoint subheaps
\ls$m_0$ (the footprint) and \ls$m_1$ (the frame), such that the
postcondition \ls$p$, when properly framed, holds over the footprint.

In order to provide specifications for primitive
actions we start in small-footprint style.
For instance, below is the WP for reading a reference:
\begin{lstlisting}
let read_wp (#a:Type) (r:ref a) = fun post m_0 -> exists x. m_0 == r |-> x /\ post x m_0
\end{lstlisting}
We then insert framing wrappers around such small-footprint WPs when
exposing the corresponding stateful actions to the programmer, e.g.,
\begin{lstlisting}
val (!) : #a:Type -> r:ref a -> STATE a (fun p m -> frame_wp (read_wp r) p m)
\end{lstlisting}
%
%
To verify code written in such style, we annotate the corresponding 
programs to have their VCs processed by \ls$sl_auto$. For instance, 
for the \ls$swap$ function below, the tactic  
successfully finds the frames for the four occurrences of
the frame rule and greatly reduces the solver's work.
%
%
Even in this simple example, not performing such instantiation would cause
the solver to fail.
\begin{lstlisting}
let swap_wp (r_1 r_2 : ref int) =
    fun p m -> exists x y. m == (r_1 |> x <*> r_2 |> y) /\ p () (r_1 |> y <*> r_2 |> x)
let swap (r_1 r_2 : ref int) : ST unit (swap_wp r_1 r_2) by (sl_auto ()) =
    let x = !r_1 in let y = !r_2 in r_1 := y; r_2 := x
\end{lstlisting}
%
%


%


The \ls$sl_auto$ tactic:
(1) uses syntax inspection to unfold and traverse the goal
until it reaches a \ls$frame_wp$---say, the one for \ls$!r_2$;
(2) inspects \ls$frame_wp$'s first explicit argument (here \ls$read_wp r_2$) to
compute the references the current command requires (here \ls$r_2$);
(3) uses unification variables to build a memory expression describing the
required framing of input memory
(here \ls{r_2$\hspace{0.01cm}$ |-> ?u_1 $\bullet$ ?u_2})
and instantiates the existentials of \ls$frame_wp$ with these unification variables;
(4) builds a goal that equates this memory expression with \ls$frame_wp$'s
third argument (here \ls{r_1$\hspace{0.01cm}$ |-> x $\bullet$ r_2$\hspace{0.01cm}$ |-> y});
and (5) uses a commutative monoids tactic
(similar to \autoref{sec:non-linear})
with the heap algebra
(\ls{emp}, \ls{$\bullet$}) to canonicalize the equality
and sort the heaplets.
Next, it can solve for the unification variables component-wise,
instantiating \ls$?u_1$ to \ls$y$ and \ls$?u_2$ to \ls$r_1 |-> x$, 
and then proceed to the next \ls{frame_wp}.
%


In general, after frames are instantiated, the SMT solver can
efficiently prove the remaining assertions, such as the obligations
about heap definedness.
Thus, with relatively little effort, \metafstar brings an
(albeit simple version of a) widely used
yet previously out-of-scope program logic (i.e., separation logic) into \fstar.
To the best of our knowledge, the ability to \emph{script} separation
logic into an SMT-based program verifier, without any primitive
support, is unique.

\subsection{Metaprogramming Verified Low-level Parsers and Serializers}
\label{sec:minilowparse}






Above, we used \metafstar to manipulate VCs for user-written code.
Here, we focus instead on generating verified code automatically.
We loosely refer to the previous setting as using ``tactics'', and to
the current one as ``metaprogramming''.
In most ITPs, tactics and
metaprogramming are not distinguished; however in a program verifier
like \fstar, where some proofs are not materialized at
all (\autoref{sec:correctness}),
proving VCs of existing terms is distinct from generating new terms.
%


Metaprogramming in \fstar involves programmatically generating a
(potentially effectful) term (e.g., by constructing its syntax and
instructing \fstar how to type-check it) and processing any VCs
that arise via tactics. When applicable (e.g., when working
in a domain-specific language), metaprogramming verified code can
substantially reduce, or even eliminate, the burden of manual proofs.

We illustrate this by automating the generation of parsers and
serializers from a type definition. Of course, this
is a routine task in many mainstream
metaprogramming frameworks (e.g., Template Haskell, camlp4, etc). The
novelty here is that we produce imperative
parsers and serializers extracted to C, with proofs that they are memory safe,
functionally correct, and mutually inverse.
This section is slightly simplified, more detail can be found
\iffull
in \autoref{sec:parsers-detail}.
\else
the appendix.
\fi

We proceed in several stages. First, we program a library of pure, high-level parser
and serializer combinators, proven to be (partial) mutual inverses of
each other.
A \ls$parser$ for a type \ls$t$
is represented as a function possibly returning a \ls$t$ along with
the amount of input bytes consumed.
The type of a \ls$serializer$ for
a given \ls$p:parser t$ contains a refinement%
\footnote{
\fstar syntax for refinements is \ls$x:t\{phi\}$, denoting
the type of all \ls$x$ of type \ls$t$ satisfying \ls$phi$.
}
stating that \ls$p$ is an
inverse of the serializer.
A \ls$package$ is a dependent record of a
parser and an associated serializer.
\begin{lstlisting}
let parser t = seq byte -> option (t * nat)
let serializer #t (p:parser t) = f:(t -> seq byte){forall x. p (f x) == Some (x, length (f x))}
type package t = { p : parser t ; s : serializer p }
\end{lstlisting}
Basic combinators in the library include constructs for parsing
and serializing base values and pairs, such as the following:
\begin{lstlisting}
val p_u8 : parse u8
val s_u8 : serializer p_u8
val p_pair : parser t1 -> parser t2 -> parser (t1 * t2)
val s_pair : serializer p1 -> serializer p2 -> serializer (p_pair p1 p2)
\end{lstlisting}
%
%
Next, we define low-level versions of these combinators, which work over
mutable arrays instead of byte sequences.
These combinators are coded in the \lowstar subset of \fstar (and so can be
extracted to C) and are proven to both be memory-safe
and respect their high-level variants.
The type for low-level parsers,
\ls$parser_impl (p:parser t)$, denotes
an imperative function that reads from an array of bytes and
returns a \ls$t$,
behaving as the specificational
parser \ls$p$.
Conversely, a \ls$serializer_impl$ \ls$(s:serializer p)$
writes into an array of bytes, behaving as \ls$s$.

Given such a library, we would like to build verified, mutually
inverse, low-level parsers and serializers for specific data
formats. The task is mechanical, yet overwhelmingly tedious
by hand, with many auxiliary proof obligations of a predictable
structure: a perfect candidate for metaprogramming.

%

\paragraph{Deriving specifications from a type definition}








Consider the following \fstar type,
representing lists of exactly $18$ pairs of bytes.
\begin{lstlisting}
type sample = nlist 18 (u8 * u8)
\end{lstlisting}
%
%
The first component of our metaprogram is \ls+gen_specs+, which
generates parser and serializer specifications from a type definition.
%
%
\begin{lstlisting}
let ps_sample : package sample = _ by (gen_specs (`sample))
\end{lstlisting}
The syntax \ls$_ by tau$ is the way to call \metafstar for code
generation. \metafstar will run the metaprogram \ls$tau$ and, if
successful, replace the underscore by the result.
In this case, the \ls$gen_specs (`sample)$ inspects the syntax of
the \ls$sample$ type (\autoref{sec:quoting})
and produces the package below (\ls$seq_p$ and \ls$seq_s$ are 
sequencing combinators):
\begin{lstlisting}
let ps_sample = { p = p_nlist 18 (p_u8 `seq_p` p_u8)
                 ; s = s_nlist 18 (s_u8 `seq_s` s_u8) }
\end{lstlisting}
%
%
\paragraph{Deriving low-level implementations that match specifications}

From this pair of specifications, we can automatically generate \lowstar
implementations for them:
\begin{lstlisting}
let p_low : parser_impl ps_sample.p = _ by gen_parser_impl
let s_low : serializer_impl ps_sample.s = _ by gen_serializer_impl
\end{lstlisting}
which will produce the following low-level implementations:
\begin{lstlisting}
let p_low = parse_nlist_impl 18ul (parse_u8_impl `seq_pi` parse_u8_impl)
let s_low = serialize_nlist_impl 18ul (serialize_u8_impl `seq_si` serialize_u8_impl)
\end{lstlisting}
For simple types like the one above, the generated code is fairly
simple. However, for more complex types, using the combinator library
comes with non-trivial proof obligations. For example, even for a
simple enumeration, \ls$type color =$ \ls$ Red$ \ls$| Green$,
the parser specification is as follows:
\begin{lstlisting}
parse_synth (parse_bounded_u8 2)
            (fun x2 -> mk_if_t (x2 = 0uy) (fun _ -> Red) (fun _ -> Green))
            (fun x -> match x with | Green -> 1uy | Red -> 0uy)
\end{lstlisting}
We represent \ls$Red$ with \ls$0uy$ and \ls$Green$ with \ls$1uy$. The
parser first parses a ``bounded'' byte, with only two
values. The \ls$parse_synth$ combinator then expects functions between
the bounded byte and the datatype being parsed (\ls$color$), which
must be proven to be mutual inverses.
This proof is conceptually easy,
but for large enumerations nested deep within the structure of other
types, it is notoriously hard for SMT solvers.
Since the proof is
inherently computational, a proof that destructs the inductive type
into its cases and then normalizes is much more natural.
With our metaprogram, we can produce the term and then discharge
these proof obligations with a tactic \emph{on the spot},
eliminating them from the final VC.
We also explore simply tweaking the SMT context, again via a tactic,
with good results.
A quantitative evaluation is provided
in~\autoref{sec:eval-minilowparse}.

\section{The Design of \metafstar}
\label{sec:design}

Having caught a glimpse of the use cases for \metafstar, we now turn to its
design.
As usual in proof assistants (such as Coq, Lean and Idris), \metafstar
tactics work over a set of goals and apply primitive actions to
transform them, possibly solving some goals and generating new goals in
the process.
Since this is standard, we will focus the most on describing the aspects
where \metafstar differs from other engines.
We first describe how metaprograms are modelled as an effect
(\autoref{sec:tac-effect})
and their runtime model (\autoref{sec:executing}).
We then detail some of \metafstar's syntax inspection and building
capabilities (\autoref{sec:quoting}).
Finally, we show how to perform some (lightweight) verification of
metaprograms (\autoref{sec:specifying-metaprograms}) within \fstar.

\subsection{An Effect for Metaprogramming}
\label{sec:tac-effect}

\metafstar tactics are, at their core, programs that transform the
``proof state'', i.e. a set of goals needing to be solved.
As in Lean~\cite{EbnerURAM17} and
Idris~\cite{ChristiansenB16}, we define a monad combining exceptions
and stateful computations over a proof state,
along with actions that can access internal components such as the
type-checker.
For this we first introduce abstract types for the proof state,
goals, terms, environments, etc., together with functions to
access them, some of them shown below.





\vspace{-5mm}
\noindent
\begin{minipage}[t]{.50\textwidth}
\begin{lstlisting}
type proofstate
type goal
type term
type env
\end{lstlisting}
\end{minipage}
\begin{minipage}[t]{.50\textwidth}
\begin{lstlisting}
val goals_of : proofstate -> list goal
val goal_env  : goal -> env
val goal_type : goal -> term
val goal_solution : goal -> term
\end{lstlisting}
\end{minipage}
%
%
%
We can now define our metaprogramming monad: \ls$tac$.
It combines \fstar's existing effect for potential divergence (\ls$Div$),
with exceptions and stateful
computations over a \ls$proofstate$.
The definition of \ls$tac$, shown below, is straightforward and given in
\fstar's standard library.
Then, we use \fstar's effect extension capabilities~\cite{dm4free}
in order to
elevate the \ls$tac$ monad and its actions to an effect, dubbed
\ls$TAC$.
 
\noindent
\begin{lstlisting}
type error = exn * proofstate (* error and proofstate at the time of failure *)
type result a = | Success : a -> proofstate -> result a | Failed  : error -> result a
let tac a = proofstate -> Div (result a)
let t_return #a (x:a) = fun ps -> Success x ps
let t_bind #a #b (m:tac a) (f:a -> tac b) : tac b = fun ps -> ... (* omitted, yet simple *)
let get () : tac proofstate = fun ps -> Success ps ps
let raise #a (e:exn) : tac a = fun ps -> Failed (e, ps)
new_effect { TAC with repr = tac ; return = t_return ; bind = t_bind
                                 ; get = get ; raise = raise }
\end{lstlisting}
%
%
The \ls$new_effect$ declaration introduces \emph{computation types} of the
form \ls$TAC t wp$, where \ls$t$ is the return type and \ls$wp$ a specification.
However, until \autoref{sec:specifying-metaprograms} we shall only use
the derived form \ls$Tac t$, where the specification is trivial.
%
%
%
These computation types are distinct from their underlying monadic
representation type \ls$tac t$---users cannot directly access the
proof state except via the actions.
The simplest actions stem from the \ls$tac$ monad definition:
\ls$get : unit -> Tac proofstate$ returns the current
proof state and
\ls$raise: exn -> Tac 'a$
fails with the given exception%
\footnote{We use greek letters $\alpha$, $\beta$, ...
to abbreviate universally quantified type variables.}.
Failures can be handled using 
\ls$catch : (unit -> Tac 'a)$ \ls$ -> Tac (either exn 'a)$,
which resets the state on failure, including that of unification
metavariables.
We emphasize two points here.
First, there is no ``\ls$set$'' action.
This is to forbid metaprograms from arbitrarily replacing their
proof state, which would be unsound.
%
%
Second, the argument to \ls$catch$ must be thunked, since in \fstar{}
impure un-suspended computations are evaluated before they are passed
into functions.\ch{Just that the CBN normalizer doesn't do this, and
  we use the CBN normalizer for tactics, which are super impure
  (because of Div and quoting observing evaluation order). What a ugly hack!}
%
%

The only aspect differentiating \ls$Tac$ from other
user-defined effects is the existence of effect-specific primitive actions,
which give access to the metaprogramming engine proper.
We list here but a few:
\begin{lstlisting}
val trivial : unit -> Tac unit $\quad$ val tc : term -> Tac term $\quad$ val dump : string -> Tac unit
\end{lstlisting}
All of these are given an interpretation internally by \metafstar.
For instance, \ls$trivial$ calls into \fstar's logical simplifier to check
whether the current goal 
is a
trivial proposition and discharges it if so, failing otherwise.
The \ls$tc$ primitive queries the type-checker
to infer the type of a given term in the current 
environment
(\fstar types are a kind of terms, hence the codomain of \ls$tc$ is also \ls$term$).
This does not change the proof state; its only purpose
is to return useful information to the calling metaprograms.
%
Finally, \ls$dump$ outputs the current proof state to the user in a
pretty-printed format, in support of user interaction.

Having introduced the \ls$Tac$ effect and some basic actions,
writing metaprograms is as straightforward as writing any other \fstar
code.
For instance, here are two metaprogram combinators.
The first one repeatedly
calls its argument until it fails, returning a list of all the
successfully-returned values. The second one behaves similarly,
but folds the results with some provided folding function.
\begin{lstlisting}
let rec repeat (tau : unit -> Tac 'a) : Tac (list 'a) =
  match catch tau with   | Inl _ -> []   | Inr x -> x :: repeat tau

let repeat_fold f e tau = fold_left f e (repeat tau)
\end{lstlisting}
 
These two small combinators illustrate a few key points
of \metafstar.
As for all other \fstar effects, metaprograms are written in
applicative style, without explicit \ls{return}, \ls{bind}, or \ls{lift}
of computations (which are inserted under the hood).
This also works across different effects: \ls$repeat_fold$ can seamlessly combine
the pure \ls$fold_left$ from \fstar's list library with a metaprogram
like \ls$repeat$.
Metaprograms are also type- and effect-inferred: while
\ls$repeat_fold$ was not at all annotated, \fstar infers the polymorphic
type
\ls$('b -> 'a -> 'b) -> 'b -> (unit -> Tac 'a) -> Tac 'a$
for it.
%

It should be noted that, if lacking an effect extension feature, one
could embed metaprograms simply via the (properly abstracted) \ls$tac$
monad instead of the \ls$Tac$ effect.\ch{As Lean actually does?}
It is just more convenient to use an effect, given we are working within
an effectful program verifier already.
%
%
In what follows, with the exception
of~\autoref{sec:specifying-metaprograms} where we describe
specifications for metaprograms, there is little reliance on
using an effect; so, the same ideas could be
applied in other settings.\ch{just that it's often not our ideas,
  and the rest have already been applied in Lean, Idris, and Mtac}
  \guido{reworded to ``the same ideas''}

\subsection{Executing \metafstar Metaprograms}
\label{sec:executing}


Running metaprograms involves three steps.
First, they are \emph{reified}~\cite{dm4free} into their underlying
\ls$tac$ representation, i.e. as state-passing functions.
%
%
%
User code cannot reify metaprograms: only \fstar
can do so when about to process a goal.
 
Second, the reified term is applied to an initial proof state, and then
simply evaluated according to \fstar's dynamic semantics, for instance
using \fstar's existing normalizer.
For intensive applications,
such as proofs by reflection,
we provide faster alternatives (\autoref{sec:evaluation}).
In order to perform this second step, the proof state, which
up until this moments exists only internally to \fstar,
must be \emph{embedded} as a term, i.e., as abstract syntax.
Here is where its abstraction pays off: since metaprograms cannot
interact with a proof state except through a limited interface, it
need not be \emph{deeply} embedded as syntax.
By simply wrapping the internal proofstate into a new kind of ``alien''
term, and making the primitives aware of this wrapping, we can readily
run the metaprogram that safely carries its alien proof state around.
This wrapping of proof states is a constant-time
operation.
%

The third step is interpreting the primitives.
They are realized by functions of similar types implemented within the
\fstar type-checker, but over an internal \ls$tac$ monad and the concrete
definitions for \ls$term$, \ls$proofstate$, etc.
Hence, there is a translation involved on every call and return,
switching between embedded representations and their
concrete variants.
Take \ls$dump$, for example, with type \ls$string -> Tac unit$.
Its internal implementation, implemented within the \fstar type-checker,
has type
\ls$string -> proofstate ->$
\ls$Div (result unit)$.
When interpreting a call to it, the interpreter must \emph{unembed} the
arguments (which are representations of \fstar terms)
into a concrete string and a concrete
proofstate to pass to the internal implementation of \ls$dump$.
The situation is symmetric for the return value of
the call, which must be \emph{embedded} as a term.

\subsection{Syntax Inspection, Generation, and Quotation}
\label{sec:quoting}

If metaprograms are to be reusable over different kinds of goals,
they must be able to reflect on the goals they are invoked to solve.
Like any metaprogramming system, \metafstar offers a way to inspect and
construct the syntax of \fstar terms.
Our representation of terms as an inductive type, and the variants
of quotations, are inspired by the ones
in Idris~\cite{ChristiansenB16} and Lean~\cite{EbnerURAM17}.
\ch{The section doesn't talk about efficiency:
``As this inspection happens very often, care must be taken to make it
efficient.''}
%
%
%

%
%


%
%

\paragraph*{\bf Inspecting syntax} Internally, \fstar uses
a locally-nameless representation~\cite{ChargueraudLocallyNameless}
with explicit, delayed
substitutions.
%
To shield metaprograms from some of this internal bureaucracy, we expose
a simplified view~\cite{WadlerViews} of terms.
Below we present a few constructors from the \ls$term_view$
type:

\vspace{-5mm}
\noindent
\begin{minipage}[t]{.42\textwidth}
\begin{lstlisting}
val inspect : term -> Tac term_view
val pack : term_view -> term
\end{lstlisting}
\end{minipage}
\begin{minipage}[t]{.58\textwidth}
\begin{lstlisting}
type term_view =
  | Tv_BVar   : v:dbvar -> term_view
  | Tv_Var    : v:name  -> term_view
  | Tv_FVar   : v:qname -> term_view
  | Tv_Abs    : bv:binder -> body:term -> term_view
  | Tv_App    : hd:term -> arg:term -> term_view
  ...
\end{lstlisting}
\end{minipage}
%
%
%
The \ls$term_view$ type provides the ``one-level-deep'' structure of a
term: metaprograms must call \ls$inspect$ to reveal the structure of the
term, one constructor at a time.
The view exposes three kinds of variables: bound variables, \ls$Tv_BVar$;
named local variables \ls$Tv_Var$; and top-level fully qualified names,
\ls$Tv_FVar$.
Bound variables and local variables are distinguished since the internal
abstract syntax is locally nameless.
For metaprogramming, it is usually simpler to use a fully-named
representation, so we provide \ls$inspect$ and \ls$pack$ functions
that open and close binders appropriately to maintain this invariant.
Since opening binders requires freshness, \ls$inspect$ has effect \ls$Tac$.%
\footnote{We also provide functions \ls$inspect_ln$, \ls$pack_ln$
which stay in a locally-nameless representation and are thus pure, total
functions.}
As generating large pieces of syntax via the view easily becomes
tedious, we also provide some ways of \emph{quoting} terms:

\paragraph*{\bf Static quotations}
A static quotation \ls$`e$ is just a shorthand for statically calling
the \fstar{} parser to convert \ls$e$ into the abstract syntax of
\fstar terms above.
For instance, \ls$`(f 1 2)$ is equivalent to the following,
\begin{lstlisting}
pack (Tv_App (pack (Tv_App (pack (Tv_FVar "f"))
                              (pack (Tv_Const (C_Int 1)))))
               (pack (Tv_Const (C_Int 2))))
\end{lstlisting}

\paragraph*{\bf Dynamic quotations}
A second form of quotation is
\ls$dquote: #a:Type -> a ->$ \ls$Tac term$, an effectful
operation that is interpreted by \fstar's normalizer during
metaprogram evaluation.
It returns the syntax of its argument at the time
\ls$dquote e$ is evaluated.
Evaluating
\ls$dquote e$ substitutes all the free variables in
\ls$e$ with their current values in the execution environment,
suspends further evaluation, and returns the abstract syntax of the
resulting term. For instance, evaluating
\ls$(fun x -> dquote (x + 1)) 16$ produces the abstract
syntax of \ls$16 + 1$.
%

\paragraph*{\bf Anti-quotations}
Static quotations are useful for building big chunks of syntax concisely,
but they are of limited use if we cannot combine them with existing bits of syntax.
Subterms of a quotation are allowed to ``escape'' and be substituted
by arbitrary expressions.
We use the syntax \ls$`#t$ to denote an antiquoted \ls$t$, where \ls$t$
must be an expression of type \ls$term$ in order for the quotation to
be well-typed.
For example, \ls$`(1 + `#e)$ creates syntax for an addition
where one operand is the integer constant \ls$1$ and the
other is the term represented by \ls$e$.
%

\paragraph*{\bf Unquotation}
Finally, we provide an effectful operation,
\ls$unquote: #a:Type ->$ \ls$t:term$ \ls$-> Tac a$, which takes a term
representation \ls$t$
and an expected type for it \ls$a$ (usually inferred
from the context), and calls the \fstar type-checker to check and
elaborate the term representation into a well-typed term.
%

\subsection{Specifying and Verifying Metaprograms}
\label{sec:specifying-metaprograms}


Since we model metaprograms as a particular kind of effectful program
within \fstar, which is a program verifier, a natural question to ask is
whether \fstar can specify and verify metaprograms.
The answer is ``yes, to a degree''.

%

To do so, we must use the WP calculus for the \ls$TAC$ effect:
\ls$TAC$-computations are given computation types of the form
\ls$TAC a wp$, where \ls$a$ is the computation's result type and
\ls$wp$ is a weakest-precondition transformer of type
\ls$tacwp a$ = \ls$proofstate -> (result a -> prop) -> prop$.
However, since WPs tend to not be very intuitive, we first define two variants
of the \ls$TAC$ effect: \ls$TacH$ in ``Hoare-style'' with pre- and
postconditions and \ls$Tac$ (which we have seen before), which
only specifies the return type, but uses trivial pre- and postconditions.
%
The \ls$requires$ and \ls$ensures$ keywords below
simply aid readability of pre- and postconditions---they are identity functions.
\begin{lstlisting}
effect TacH (a:Type) (pre : proofstate -> prop) (post : proofstate -> result a -> prop) =
        TAC a (fun ps post' -> pre ps /\ (forall r. post ps r ==> post' r))
effect Tac (a:Type) = TacH a (requires (fun _ -> True)) (ensures (fun _ _ -> True))
\end{lstlisting}
Previously, we only showed the simple type for the \ls$raise$ primitive,
namely \ls$exn ->$ \ls$Tac 'a$.
In fact, in full detail and Hoare style, its type/specification is:
\begin{lstlisting}
val raise : e:exn-> TacH 'a $\,\!$ (requires (fun _ -> True))
                          (ensures (fun ps r -> r == Failed (e, ps)))
\end{lstlisting}
expressing that the primitive has no precondition, always fails with
the provided exception, and does not modify the proof state.
From the specifications of the primitives,
and the automatically obtained Dijkstra monad,
\fstar can already prove interesting properties about
metaprograms.
We show a few simple examples.

The following metaprogram is accepted by \fstar as it can conclude,
from the type of \ls$raise$, that the assertion is unreachable, and
hence \ls$raise_flow$ can have a trivial precondition (as \ls$Tac unit$
implies).
\begin{lstlisting}
let raise_flow () : Tac unit = raise SomeExn; assert False
\end{lstlisting}
For \ls$cur_goal_safe$ below, \fstar verifies that (given the precondition)
the pattern match is exhaustive.
The postcondition is also asserting that the metaprogram always succeeds
without affecting the proof state, returning some unspecified goal.
Calls to \ls$cur_goal_safe$ must statically ensure that the goal list is
not empty.
\begin{lstlisting}
let cur_goal_safe () : TacH goal (requires (fun ps -> ~(goals_of ps == [])))
                                (ensures (fun ps r -> exists g. r == Success g ps)) =
    match goals_of (get ()) with | g :: _ -> g
\end{lstlisting}
Finally, the \ls$divide$ combinator below ``splits'' the goals of a
proof state in two at a given index \ls$n$, and focuses a different
metaprogram on each.
It includes a runtime check that the given \ls$n$ is non-negative,
and raises an exception in the \ls$TAC$ effect otherwise.
Afterwards, the call to the (pure) \ls$List.splitAt$ function
requires that \ls$n$ be statically known to be non-negative,
a fact which can be proven from the specification for
\ls$raise$ and the effect definition, which defines
the control flow.
%
\begin{lstlisting}
let divide (n:int) (tl : unit -> Tac 'a) (tr : unit -> Tac 'b) : Tac ('a * 'b) =
    if n < 0 then raise NegativeN;
    let gsl, gsr = List.splitAt n (goals ()) in ...
\end{lstlisting}
This enables a style of ``lightweight'' verification of metaprograms,
where expressive invariants about their state and control-flow can be
encoded.
The programmer can exploit dynamic checks (\ls$n < 0$) and exceptions
(\ls$raise$) or static ones (preconditions), or a mixture of them, as
needed.

Due to type abstraction, though, the specifications of most primitives cannot provide
complete detail about their behavior, and deeper specifications (such as
ensuring a tactic will correctly solve a goal) cannot currently be proven,
nor even stated---to do so would require, at least, an internalization
of the typing judgment of \fstar.
While this is an exciting possibility~\cite{AnandBCST18},
we have for now only focused on
verifying basic safety properties of metaprograms, which helps users
detect errors early, and whose proofs the SMT can handle well.
Although in principle, one can also write tactics to
discharge the proof obligations of metaprograms.

\section{\metafstar, Formally}

We now describe the trust assumptions for \metafstar
(\autoref{sec:correctness}) and then how we reconcile tactics within a
program verifier, where the exact shape of VCs is not given, nor known
a priori by the user (\autoref{sec:vcgen}).


\subsection{Correctness and Trusted Computing Base (TCB)}
\label{sec:correctness}

As in any proof assistant,
tactics and metaprogramming would be rather useless if
they allowed to ``prove'' invalid judgments---care must be taken to ensure soundness.
%
We begin with a taste of the specifics of \fstar's static
semantics, which influence the trust model for \metafstar, and then
provide more detail on the TCB.

\paragraph{\bf Proof irrelevance in \fstar}
The following two rules for introducing and eliminating refinement types
are key in \fstar, as they form the basis of its
proof irrelevance.

\[
   \inferrule*[lab=T-Refine]
   {
     \Gamma |- e ~:~ t \quad
     \Gamma |= \phi[e/x]
   }
   {
     \Gamma |- e ~:~ x\!:\!t\{\phi\}
   }
   \qquad
   \qquad
   \inferrule*[lab=V-Refine]
   {
     \Gamma |- e ~:~ x\!:\!t\{\phi\}
   }
   {
     \Gamma |= \phi[e/x]
   }
\]

The $\vDash$ symbol represents \fstar's \emph{validity judgment}~\cite{dm4free}
which, at a high-level, defines a proof-irrelevant,
classical, higher-order logic.
These validity hypotheses are usually collected by the type-checker, and
then encoded to the SMT solver in bulk.
Crucially, the irrelevance of validity is what permits
efficient interaction with SMT solvers, since
reconstructing \fstar terms from SMT proofs is unneeded.

As evidenced in the rules, validity and typing are mutually recursive,
and therefore \metafstar must also construct validity derivations.
In the implementation, we model these validity goals as
holes with a ``squash'' type~\citep{Nogin02,AwodeyBauer},
where \ls$squash phi = _:unit{phi}$, i.e., a refinement of \ls$unit$.
Concretely, we model \VGoal{\Gamma}{phi} as \Goal{\Gamma}{?u}{squash phi} using 
a unification variable.
\metafstar does not construct deep solutions to squashed goals: if they
are proven valid, the variable \ls$?u$ is simply solved by the \ls{unit} value `\ls$()$'.
At any point, any such irrelevant goal can be sent to the SMT solver.
Relevant goals, on the other hand, cannot be sent to SMT.




\paragraph{\bf Scripting the typing judgment}
A consequence of validity proofs not being materialized
is that type-checking is undecidable in \fstar.
For instance: does the unit value \ls$()$ solve the hole
\Goal{\Gamma}{?u}{squash phi}?
Well, only if $\phi$ holds---a condition which no type-checker can effectively
decide.
This implies that the type-checker cannot, in general, rely on proof
terms to reconstruct a proof.\ch{Why? I have no clue why this is related
  to decidability. Isn't higher-order unification in Coq already undecidable?
  Does this mean that Coq cannot rely on proof terms?}
Hence, the primitives are designed to provide access to the typing
judgment of \fstar directly, instead of building syntax for proof terms.
One can think of \fstar's type-checker as
implementing one particular algorithmic heuristic of the typing
and validity judgments---a heuristic which happens to work well in practice.
For convenience, this default type-checking heuristic is also available to
metaprograms: this is in fact precisely what the \ls$exact$ primitive
does.
Having programmatic access to the typing judgment
also provides the flexibility to tweak VC generation as needed,
instead of leaving it to the default behavior of \fstar.
For instance, the \ls$refine_intro$ primitive implements
\textsc{T-Refine}.
When applied, it produces two new goals, including that the refinement
actually holds.
At that point, a metaprogram can run any arbitrary tactic on it,
instead of letting the \fstar type-checker collect the obligation
and send it to the SMT solver in bulk with others.




\paragraph{\bf Trust}
\label{sec:trust}

There are two common approaches for the correctness of tactic engines:
(1) the \emph{de Bruijn
criterion}~\cite{Barendregt:2001:PUD:778522.778527}, which
requires constructing full proofs (or proof terms)
and checking them at the end, hence reducing trust to
an independent proof-checker; 
and (2) the LCF style, which applies backwards reasoning while
constructing validation functions at every
step, reducing trust to primitive, forward-style implementations of the
system's inference rules.

As we wish to make use of SMT solvers
within \fstar, the first approach is not easy.
Reconstructing the proofs SMT solvers produce, if any, back into a
proper derivation remains a significant challenge (even despite recent
progress, e.g.~\cite{BohmeW10,smtcoqCAV2017}).
Further, the logical encoding from \fstar to
SMT, along with the solver itself, are already part of \fstar's
TCB: shielding \metafstar from them would not significantly increase
safety of the combined system.

Instead, we roughly follow
the LCF approach and implement \fstar's typing rules
as the basic user-facing metaprogramming actions.
However, instead of implementing the rules in forward-style and using
them to validate (untrusted) backwards-style tactics, we implement them
directly in backwards-style.
That is, they run by breaking down goals into subgoals,
instead of combining proven facts into new proven facts.
%
Using LCF style makes the primitives part of the TCB.
However, given the primitives are sound, any combination of them
also is, and any user-provided metaprogram must be safe due to the
abstraction imposed by the \ls$Tac$ effect, as discussed next.

\subsubsection{Correct Evolutions of the Proof State}

For soundness, it is imperative that tactics do not
arbitrarily drop goals from the proof state,
and only discharge them when they are solved,
or when they can be solved by other goals tracked in the proof state.
For a concrete example, consider the following program:
\begin{lstlisting}
let f : int -> int = _ by (intro (); exact (`42))
\end{lstlisting}
Here, \metafstar will create an initial proof state with
a single goal of the form
$\left[\Goal{\emptyset}{?u_1}{int -> int}\right]$ and
begin executing the metaprogram.
When applying the \ls$intro$ primitive, the proof state
transitions as shown below.
%
%
\[
    \left[ \Goal{\emptyset}{?u_1}{int -> int} \right]
    \leadsto
    \left[ \Goal{\kw{x:int}}{?u_2}{int} \right]
\]
\let\evol\preceq
Here, a solution to the original goal
has not yet been built, since it
\emph{depends} on the solution to the goal on the right hand side.
When it is solved with, say, \ls$42$, we can solve our original goal
with \ls$fun x -> 42$.
To formalize these dependencies, we say that a proof state $\phi$
\emph{correctly evolves (via $f$) to $\psi$}, denoted $\phi \evol_f \psi$,
when there is a generic transformation $f$, called a \emph{validation}, from solutions to all of $\psi$'s
goals into correct solutions for $\phi$'s goals.
When $\phi$ has $n$ goals and $\psi$ has $m$ goals, the
validation $f$ is a function from $\kw{term}^m$ into $\kw{term}^n$.
%
Validations may be composed, providing the
transitivity of correct evolution, and if a proof state $\phi$
correctly evolves (in any amount of steps) into a state with no more goals, then we have fully
defined solutions to all of $\phi$'s goals.
We emphasize that validations are not constructed
explicitly during the execution of metaprograms.
Instead we exploit unification metavariables to instantiate the
solutions automatically.

%


Note that
validations may construct solutions for more than one goal, i.e.,
their codomain is not a single term.
This is required in \metafstar, where primitive steps may not only
decompose goals into subgoals, but actually combine goals as well.
Currently, the only primitive providing this behavior is \ls$join$,
which finds a maximal common prefix of the environment of two
irrelevant goals, reverts the ``extra'' binders in both goals and
builds their conjunction.
Combining goals using \ls$join$ is especially useful for
sending multiple goals to the SMT solver in a single call.
When there are common obligations within two goals, \ls$join$ing them
before calling the SMT solver can result in a significantly faster
proof.

\iffull
The situation when this proves useful is in sending the combined goal to
SMT.
The bulk of an SMT solver's work is extending its set of facts
by using logical reasoning and triggering quantifiers via E-matching,
the latter being relatively expensive.
Once a fact enters the set, it is very efficient to find its
proof\guido{proof? yikes} which essentially means any particular fact
needs to only be proven once.
In the where there are common obligations within two goals, \ls$join$ing
them before calling the SMT solver can result in a significantly faster
proof.
Joining everything, however, can degrade performance and predictability.
Another alternative is, of course, using tactics to cut via the common
obligation, obtaining a separate goal for it and allowing to assume it
elsewhere.
While this works, it requires the programmer to identify the formula
and explicitly cut by it, or write an automated procedure to find these
formulas.
With \metafstar, we can simply reuse the solver's good support for
identifying repeats.
\fi

We check that every primitive action respects the $\evol$ preorder.
This relies on them modeling \fstar's typing rules.
For example, and unsurprisingly, the following rule for typing abstractions
is what justifies the \ls$intro$ primitive:
%
\[
   \inferrule*[lab=T-Fun]
   {
     \Gamma, x:t |- e : t'
   }
   {
     \Gamma |- \lambda(x:t).e ~:~ (x:t) -> t'
   }
\]
Then, for the proof state evolution above,
the validation function $f$
is the (mathematical, meta-level) function taking a term of type
$int$ (the solution for \ls$?u_2$) and building syntax for its abstraction over $x$.
Further, the \ls$intro$ primitive
respects the correct-evolution
preorder, by the very typing rule (T-Fun) from which it is defined.
In this manner, every typing rule induces a
syntax-building metaprogramming step.
Our primitives come from this dual interpretation of typing rules,
which ensures that logical consistency is preserved.


Since the $\evol$ relation is a preorder, and every metaprogramming
primitive we provide the user evolves the proof state according $\evol$, it is trivially the
case that the final proof state returned by a (successful)
computation is a correct evolution of the initial one.
That means that when the metaprogram terminates, one
has indeed broken down the proof obligation correctly,
and is left with a
(hopefully) simpler set of obligations to fulfill.
%
%
%
%
Note that since $\evol$ is a preorder, \ls$Tac$
provides an interesting example of monotonic state~\cite{preorders}.


\subsection{Extracting Individual Assertions}
\label{sec:vcgen}

As discussed, the logical context of a goal processed by a
tactic is not always syntactically evident in the program.
%
And, as shown in the \ls$List.splitAt$ call in \ls$divide$ from
\autoref{sec:specifying-metaprograms}, some
obligations crucially depend on the control-flow of the program.
Hence, the proof state must crucially include these assumptions
if proving the assertion is to succeed.
%
Below, we describe how \metafstar finds proper contexts in which
to prove the assertions, including control-flow information.
Notably, this process is defined over logical formulae and does not
depend at all on \fstar's WP calculus or VC generator: we believe it
should be applicable to any VC generator.

As seen in \autoref{sec:non-linear}, the basic mechanism by which
\metafstar attaches a tactic to a specific sub-goal is
\ls$assert phi by tau$. Our encoding of this expression is built
similarly to \fstar's existing \ls$assert$ construct, which is simply
sugar for a pure function \ls$_assert$ of type
\ls$phi:prop -> Lemma (requires phi) (ensures phi)$, which
essentially introduces a cut in the generated VC.
That is, the term \ls$(assert phi; e)$ roughly produces the verification condition
\ls$phi /\ (phi ==> VC_e)$, requiring a proof of \ls$phi$ at this point,
and assuming \ls$phi$ in the continuation.
%
For \metafstar, we aim to keep this style while allowing asserted
formulae to be decorated with user-provided tactics that are tasked with
proving or pre-processing them. We do this in three steps.

First, we define the following ``phantom'' predicate:
%
\begin{lstlisting}
let with_tactic (phi : prop) (tau : unit -> Tac unit) = phi
\end{lstlisting}
Here \ls$phi `with_tactic` tau$ simply associates the tactic \ls$tau$
with \ls$phi$, and is equivalent to \ls$phi$ by its definition.
%
Next, we implement the \ls$assert_by_tactic$ lemma, and desugar
\ls$assert phi by tau$ into \ls$assert_by_tactic phi$\! \ls$tau$.
This lemma is trivially provable by \fstar. 
%
\begin{lstlisting}
let assert_by_tactic (phi : prop) (tau : unit -> Tac unit)
                             : Lemma (requires (phi `with_tactic` tau)) (ensures phi) = ()
\end{lstlisting}
Given this specification,
the term \ls$(assert phi by tau; e)$ roughly produces the verification
condition
\ls$phi `with_tactic` tau /\ (phi ==> VC_e)$,
with a tagged left sub-goal,
and \ls$phi$ as an hypothesis
in the right one.
Importantly, \fstar keeps the \ls$with_tactic$ marker uninterpreted until the VC
needs to be discharged.
At that point, it may contain several annotated subformulae.
For example, suppose the VC is \ls$VC0$ below, where we distinguish
an ambient context of variables and hypotheses $\Delta$:
\begin{lstlisting}
(VC0)   $\Delta \models$ X ==> (forall (x:t). R `with_tactic` tau$_{\!1}$ /\ (R ==> S))
\end{lstlisting}
%
%
%
In order to run the \ls$tau$$_{\!1}$ tactic on \ls$R$, it must first be ``split
out''.
To do so, all logical information ``visible'' for \ls$tau$$_{\!1}$
(i.e. the set of premises of the implications traversed and
the binders introduced by quantifiers) must be included.
%
%
As for any program verifier, these hypotheses include the control flow
information, postconditions, and any other logical fact that is known to
be valid at the program point where the corresponding \ls$assert R by$\!
\ls$tau$$_{\!1}$ was called.
%
All of them are collected into $\Delta$ as the term is traversed.
In this case, the VC for $R$ is:
\begin{lstlisting}
(VC1)   $\Delta$, _:X, x:t $\models$ R
\end{lstlisting}
Afterwards, this obligation is removed from the original VC. This is done
by replacing it with \ls$True$, leaving a ``skeleton'' VC with all
remaining facts.
\begin{lstlisting}
(VC2)   $\Delta \models$ X ==> (forall (x:t). True$\hspace{0.005cm}$ /\ (R ==> S))
\end{lstlisting}
The validity of \ls$VC1$ and \ls$VC2$ implies that of \ls$VC0$.
\fstar also recursively descends into \ls$R$ and \ls$S$, in case there are more
\ls$with_tactic$ markers in them.
Then, tactics are run on the the split VCs (e.g.,
\ls$tau$$_{\!1}$ on \ls$VC1$) to break them down (or solve them).
All remaining goals, including the skeleton, are sent to the SMT solver.

Note that while the \emph{obligation} to prove \ls$R$, in \ls$VC1$, is
preprocessed by the tactic \ls$tau$$_{\!1}$, the \emph{assumption} \ls$R$ for the
continuation of the code, in \ls$VC2$, is left as-is.
This is crucial for tactics such as the canonicalizer from
\autoref{sec:non-linear}: if the skeleton \ls$VC2$
contained an assumption for the canonicalized equality
it would not help the SMT solver show the uncanonicalized postcondition.

However, not all nodes marked with \ls$with_tactic$
are proof obligations.
Suppose \ls$X$ in the previous VC was given as \ls$(Y `with_tactic` tau$$_{\!2}$\ls$)$.
In this case, one certainly does not want to attempt to prove
\ls$Y$, since it is an hypothesis.
While it would be \emph{sound} to prove it and replace
it by \ls$True$, it is useless at best,
and usually irreparably affects the system. Consider asserting the tautology
\ls$(False `with_tactic` tau) ==> False$.

Hence, \fstar splits such obligations only in strictly-positive
positions.\ch{what does strictly positive mean for logical formulas? I
  only know positive and negative. Strictly-positive is for
  inductives.}\nik{it just means, as also in the case of inductives,
  not to the left of any arrow/implication}\ch{Should we explain
  why we need the strict?}
On all others, \fstar simply drops the \ls$with_tactic$ marker, e.g., by just unfolding the definition of \ls$with_tactic$.
For regular uses of the \ls$assert..by$ construct, however, all occurrences are
strictly-positive.
It is only when (expert) users use the \ls$with_tactic$ marker directly
that the above discussion might become relevant.

Formally, the soundness of this whole approach is given by the following
metatheorem, which justifies the splitting out of sub-assertions,
and by the correctness of evolution detailed in \autoref{sec:correctness}.
The proof of \autoref{thm:split} is straightforward, and included in
\iffull
\autoref{app:correctness}.
\else
the appendix.
\fi
We expect an analogous property
to hold in other verifiers as well (in particular, it
holds for first-order logic).

\begin{theorem}
    \label{thm:split}
    Let $E$ be a context with 
    $\Gamma \vdash E : prop \Rightarrow prop$,
    and $\phi$ a squashed proposition such that $\Gamma \vdash \phi : prop$.
    Then the following holds:
    \[
        \infer{\Gamma \vDash E[\phi]}
              {\Gamma \vDash E[\top] & \Gamma, \gamma(E) \vDash \phi}
    \]
    where $\gamma(E)$ is the set of binders $E$ introduces.
    If $E$ is strictly-positive, then the reverse implication holds as well.
\end{theorem}

\section{Executing Metaprograms Efficiently}
\label{sec:evaluation}


\fstar provides three complementary mechanisms for running metaprograms.
The first two, \fstar's call-by-name (CBN)
interpreter and a (newly implemented) call-by-value (CBV)
NbE-based evaluator, support strong
reduction---henceforth we refer to these as ``normalizers''.
%
In addition, we design and implement a new \emph{native plugin}
mechanism that allows both normalizers to interface with \metafstar
programs extracted to OCaml, reusing \fstar's existing extraction pipeline for
this purpose.
%
%
Below we provide
a brief overview of the three mechanisms.

\subsection{CBN and CBV Strong Reductions}

As described in \autoref{sec:tac-effect}, metaprograms, once reified,
are simply \fstar terms of type \ls$proofstate -> Div (result a)$.
As such, they can be reduced using \fstar's existing computation
machinery, a CBN interpreter for strong reductions based on the Krivine
abstract machine (KAM)~\cite{cregut2007strongly,Krivine2007}.
Although complete and highly configurable, \fstar's KAM interpreter
is slow, designed primarily for converting types during dependent
type-checking and higher-order unification.

Shifting focus to long-running metaprograms, such as
tactics for proofs by reflection,
we implemented an NbE-based strong-reduction evaluator for \fstar
computations. The evaluator is implemented in \fstar and extracted to
OCaml (as is the rest of \fstar{}), thereby inheriting CBV from OCaml.
It is similar to~\citepos{BoespflugDG11} NbE-based
strong-reduction for Coq,
although we do not implement their low-level,
OCaml-specific tag-elimination optimizations---nevertheless, 
it is already vastly more efficient than the KAM-based interpreter.
%
%
%
%
%

\ch{Nobody worried that changing between CBN and CBV at effect Div
  does change the semantics?}

\subsection{Native Plugins \& Multi-language Interoperability}


Since \metafstar programs are just \fstar programs, they can also be
extracted to OCaml and natively compiled.
Further, they can be dynamically linked into \fstar as ``plugins''.
Plugins can be directly called from the type-checker, as is done for the
primitives, which is much more efficient than interpreting them.
%
%
However, compilation has a cost, and it is not convenient to compile
every single invocation.
Instead, \metafstar enables users to choose which metaprograms
are to be plugins
(presumably those expected to be computation-intensive, e.g.
\ls$canon_semiring$).
Users can choose their native plugins, while still quickly
scripting their higher-level logic in the interpreter.

This requires (for higher-order metaprograms) a form of multi-language
interoperability, converting between representations of terms used in
the normalizers and in native code.
%
%
We designed a small multi-language calculus, with ML-style polymorphism,
to model the interaction between normalizers and plugins and
conversions between terms.
\iffull
We outline it in
\autoref{app:native}.
\else
See the appendix for details.
\fi

Beyond the notable efficiency gains of running compiled code vs.
interpreting it, native metaprograms also require fewer embeddings.
%
Once compiled, metaprograms work over the internal, \emph{concrete}
types for \ls$proofstate$, \ls$term$, etc., instead of over their
\fstar representations (though still treating them abstractly).
Hence, compiled metaprograms can call primitives without needing to
embed their arguments or unembed their results.
Further, they can call each other directly as well.
Indeed, operationally there is little operational difference
between a primitive and a compiled metaprogram used as a plugin.

Native plugins, however, are not a replacement for the normalizers,
for several reasons.
First, the overhead in compilation might not
be justified by the execution speed-up.
Second, extraction to OCaml erases types and proofs.
As a result, the \fstar \emph{interface} of the native plugins can only
contain types that can also be expressed in OCaml, thereby excluding
full-dependent types---internally, however, they can be dependently typed.
%
Third, being OCaml programs, native plugins do not support reducing
open terms, which is often required.
However, when the programs treat their open arguments parametrically,
relying on parametric polymorphism, the normalizers can pass such arguments
\emph{as-is}, thereby recovering open reductions in some cases.
This allows us to
use native datastructure implementations (e.g. \ls{List}),
which is much faster than using the normalizers, even for open terms.
\iffull
We discuss this briefly in
\autoref{app:native}.
\else
See the appendix for details.
\fi


%

%
%

\newcommand{\bnfalt}{{\bf \,\,\mid\,\,}}
\newcommand{\bnfdef}{{\bf ::=~}}

\newcommand\eqdef{\mathrel{\overset{\makebox[0pt]{\mbox{\normalfont\tiny\sffamily def}}}{=}}}

\newcommand{\lk}[1]{%
  \texttt{%
    \fontdimen2\font=2pt   
    #1}}

\definecolor{dkblue}{rgb}{0,0.1,0.5}
\definecolor{dkgreen}{rgb}{0,0.5,0}
\definecolor{dkred}{rgb}{0.7,0,0}
\definecolor{dkpurple}{rgb}{0.7,0,0.4}

\newcommand{\highlight}[1]{%
  \colorbox{dkpurple}{$\displaystyle#1$}}

\newcommand{\blue}[1]{{\textcolor{dkblue}{#1}}}
\newcommand{\red}[1]{{\textcolor{dkpurple}{#1}}}

\newcommand{\source}[1]{\blue{#1}}
\newcommand{\sourcelk}[1]{\blue{\lk{#1}}}
\newcommand{\sexp}{\source{e}}
\newcommand{\sexpp}{\source{e'}}
\newcommand{\sexpi}{\source{e_1}}
\newcommand{\sexpii}{\source{e_2}}
\newcommand{\sexpip}{\source{e_1'}}
\newcommand{\sexpiip}{\source{e_2'}}

\newcommand{\sval}{\source{v}}
\newcommand{\svali}{\source{v_1}}
\newcommand{\svalii}{\source{v_2}}
\newcommand{\sctxeq}{\cong_{ctx}}
\newcommand{\stau}{\source{\tau}}
\newcommand{\ssig}{\source{\sigma}}
\newcommand{\staui}{\source{\tau_1}}
\newcommand{\stauii}{\source{\tau_2}}
\newcommand{\staup}{\source{\tau'}}
\newcommand{\sdelta}{\source{\delta}}
\newcommand{\sdeltai}{\source{\delta_1}}
\newcommand{\sdeltaii}{\source{\delta_2}}
\newcommand{\sdeltap}{\source{\delta'}}
\newcommand{\stunit}{\sourcelk{unit}}
\newcommand{\snat}{\sourcelk{nat}}
\newcommand{\sttop}{\source{\top}}
\newcommand{\sarrow}[2]{#1\source{\rightarrow}#2}
\newcommand{\sprod}[2]{#1\source{\times}#2}
\newcommand{\ssum}[2]{#1\source{+}#2}
\newcommand{\sall}[2]{\source{\forall}#1\source{.}#2}

\newcommand{\senv}{\source{\Gamma}}
\newcommand{\stenv}{\source{\Delta}}
\newcommand{\senvemp}{\source{\cdot}}
\newcommand{\shastyp}[2]{#1 : #2}
\newcommand{\senvext}[3]{#1,\shastyp{#2}{#3}}
\newcommand{\senvtext}[2]{#1,#2}

\newcommand{\tenv}{\target{\Gamma}}
\newcommand{\tenvemp}{\target{\cdot}}
\newcommand{\tenvext}[2]{#1,#2}

\newcommand{\sesubst}{\source{\delta}}
\newcommand{\sesubstemp}{\source{\cdot}}
\newcommand{\sesubstsingl}[2]{\source{[}{#1}\source{\mapsto}{#2}\source{]}}
\newcommand{\sesubsttwo}[4]{\source{[}{#1}\source{\mapsto}{#2},{#3}\source{\mapsto}{#4}\source{]}}
\newcommand{\sesubstlist}[1]{\source{[}{#1}\source{]}}
\newcommand{\bind}[2]{{#1}\source{\mapsto}{#2}}
\newcommand{\sesubstext}[3]{#1\source{[}{#2}\source{\mapsto}{#3}\source{]}}

\newcommand{\ssubst}[2]{\source{[}#2\source{/}#1\source{]}}
\newcommand{\tsubst}[2]{\target{[}#2\target{/}#1\target{]}}
\newcommand{\tesubstsingl}[2]{\target{[}{#1}\target{\mapsto}{#2}\target{]}}

\newcommand{\sunit}{\sourcelk{()}}
\newcommand{\svar}[1]{\source{#1}}
\newcommand{\sabs}[3]{\source{\lambda} #1 \source{:} #2 \source{.} #3}
\newcommand{\sapp}[2]{#1~#2}
\newcommand{\stabs}[2]{\source{\Lambda #1\source{.} #2}}
\newcommand{\stapp}[2]{#1~#2}
\newcommand{\spair}[2]{\source{(}#1 \source{,} #2\source{)}}
\newcommand{\sproj}[2]{#1\source{.}#2}
\newcommand{\sproji}[1]{#1\source{.i}}
\newcommand{\sfst}[1]{\sproj{#1}{\source{1}}}
\newcommand{\ssnd}[1]{\sproj{#1}{\source{2}}}
\newcommand{\sinl}[1]{\sourcelk{inl}~#1}
\newcommand{\sinr}[1]{\sourcelk{inr}~#1}
\newcommand{\scasex}[5]{\sourcelk{case}\source{(} #1  \source{,x:#2.} #3  \source{,x:#4.} #5 \source{)}}
\newcommand{\scasey}[5]{\sourcelk{case}\source{(} #1  \source{,y:#2.} #3  \source{,y:#4.} #5 \source{)}}
\newcommand{\scase}[5]{\sourcelk{case}\source{(} #1  \source{,} #2 \source{.} #3  \source{,} #4 \source{.} #5 \source{)}}
\newcommand{\salien}[3]{\source{\{}#1 \source{\}}_{#2}^{#3}}
\newcommand{\sopaque}[1]{\source{\langle}#1 \source{\rangle}}
\newcommand{\spar}[1]{\source{(}#1\source{)}} 

\newcommand{\target}[1]{\red{#1}}
\newcommand{\targetlk}[1]{\red{\lk{#1}}}

\newcommand{\texp}{\target{e}}
\newcommand{\texpp}{\target{e'}}
\newcommand{\tval}{\target{v}}
\newcommand{\tvali}{\target{v_1}}
\newcommand{\tvalii}{\target{v_2}}
\newcommand{\texpi}{\target{e_1}}
\newcommand{\texpii}{\target{e_2}}
\newcommand{\texpip}{\target{e_1'}}
\newcommand{\texpiip}{\target{e_2'}}

\newcommand{\tunit}{\targetlk{()}}
\newcommand{\tvar}[1]{\target{#1}}
\newcommand{\tabs}[2]{\target{\lambda} #1 \target{.} #2}
\newcommand{\tapp}[2]{#1~#2}
\newcommand{\tpair}[2]{\target{(}#1 \target{,} #2\target{)}}
\newcommand{\tproj}[2]{#1\target{.}#2}
\newcommand{\tproji}[1]{#1\target{.i}}
\newcommand{\tfst}[1]{\sproj{#1}{\target{1}}}
\newcommand{\tsnd}[1]{\sproj{#1}{\target{2}}}
\newcommand{\tinl}[1]{\targetlk{inl}~#1}
\newcommand{\tinr}[1]{\targetlk{inr}~#1}
\newcommand{\tcasex}[3]{\targetlk{case}\target{(}#1  \target{,x.} #2  \target{,x.} #3 \target{)}}
\newcommand{\tcasey}[3]{\targetlk{case}\target{(}#1  \target{,y.} #2  \target{,y.} #3 \target{)}}
\newcommand{\tcase}[6]{\targetlk{case}\target{(}  \target{,} #2 \target{.} #3  \target{,} #4 \target{.} #5 \target{)}}
\newcommand{\talien}[3]{\target{[}#1 \target{]}_{#2}^{#3}}
\newcommand{\topaque}[1]{\target{\langle}#1 \target{\rangle}}
\newcommand{\tpar}[1]{\target{(}#1\target{)}} 

\newcommand{\iswf}[2]{#1 \vdash #2 ~\lk{wf}}
\newcommand{\hastyp}[4]{#1;~#2 \vdash #3 : #4}
\newcommand{\isok}[3]{#1;~#2 \vdash #3 ~\lk{ok}}

\newcommand{\kap}[4]{\mathcal{K}_{#1}^{#2}(#3,#4)}
\newcommand{\mfun}[2]{\lk{fun}~#1\Rightarrow#2}
\newcommand{\tcoerce}[3]{\mathcal{\target{C}}_{#1}^{#2}(#3)}
\newcommand{\scoerce}[3]{\mathcal{\source{C}}_{#1}^{#2}(#3)}
\newcommand{\prot}[3]{\mathcal{P}_{#2}^{#3}(#1)}

\newcommand{\seval}{\source{~\rightarrow~}}
\newcommand{\sevals}{\source{~\rightarrow^{*}~}}
\newcommand{\sevalm}[1]{\source{~\rightarrow}^{#1}~}
\newcommand{\teval}{\target{~\rightarrow~}}
\newcommand{\tevals}{\target{~\rightarrow^{*}~}}
\newcommand{\tevalm}[1]{\target{~\rightarrow}^{#1}~}
\newcommand{\sevalr}[1]{\source{~\xrightarrow{\textsc{#1}}~}}
\newcommand{\tevalr}[1]{\target{~\xrightarrow{\textsc{#1}}~}}
\newcommand{\sevalrs}[1]{\source{~\xrightarrow{\textsc{#1}}^{*}~}}
\newcommand{\tevalrs}[1]{\target{~\xrightarrow{\textsc{#1}}^{*}~}}
\newcommand{\isval}[1]{\lk{is\_val}~#1}
\newcommand{\scon}[1]{\source{E[}#1\source{]}}
\newcommand{\tcon}[1]{\target{E[}#1\target{]}}

\newcommand{\trans}{~\source{\rightsquigarrow}~}

\newcommand{\sdup}{\stabs{\svar{B}}{\sabs{\svar{x}}{\svar{B}}{\spair{\svar{x}}{\svar{x}}}}}
\newcommand{\sduptyp}{\sall{\svar{B}}{\sarrow{\svar{B}}{\sprod{\svar{B}}{\svar{B}}}}}
\newcommand{\sapptyp}{\sall{\svar{A}}{\sarrow{\spar{\sduptyp}}{\sarrow{\svar{A}}{\sprod{\svar{A}}{\svar{A}}}}}}
\newcommand{\sapptypfree}{\sarrow{\spar{\sduptyp}}{\sarrow{\svar{A}}{\sprod{\svar{A}}{\svar{A}}}}}

\newcommand{\sarith}{\source{(1 + 2 + \dots)}}
\newcommand{\sinctyp}{\sarrow{\snat}{\snat}}

\newcommand{\sapptypalt}
{\sall{\svar{A}}
  {\sall{\svar{B}}
    {\sarrow{\spar{\sarrow{\svar{A}}{\svar{B}}}}{\sarrow{\svar{A}}{\svar{B}}}}}}
\newcommand{\sapptypaltf}
{\sall{\svar{B}}
     {\sarrow{\spar{\sarrow{\svar{A}}{\svar{B}}}}{\sarrow{\svar{A}}{\svar{B}}}}}
\newcommand{\sapptypaltff}
   {\sarrow{\spar{\sarrow{\svar{A}}{\svar{B}}}}{\sarrow{\svar{A}}{\svar{B}}}}

\newcommand{\sidf}[1]{\sabs{\svar{y}}{#1}{\svar{y}}}

\newcommand{\sappid}
{ \stabs
    {\svar{A}}
    {\sabs{\svar{x}}{\svar{A}}
    {\sabs{\svar{f}}
      {\sarrow{\svar{A}}{\svar{A}}}
      {{\sapp{\svar{f}}{\spar{\sapp{\spar{\sidf{\svar{A}}}}{\svar{x}}}}}}
    }}
}
\newcommand{\sappidf}
{\sabs{\svar{x}}{\svar{A}}
  {\sabs{\svar{f}}
      {\sarrow{\svar{A}}{\svar{A}}}
      {{\sapp{\svar{f}}{\spar{\sapp{\spar{\sidf{\svar{A}}}}{\svar{x}}}}}}}}
\newcommand{\sappidarg}[1]
{\sabs{\svar{f}}
      {\sarrow{\svar{A}}{\svar{A}}}
      {{\sapp{\svar{f}}{\spar{\sapp{\spar{\sidf{\svar{A}}}}{#1}}}}}}
\newcommand{\sappidargs}[2]
{\sapp{#2}{\spar{\sapp{\spar{\sidf{\svar{A}}}}{#1}}}}

\newcommand{\sappidtyp}
  {\sall{\svar{A}}
      {\sarrow{\svar{A}}{\sarrow{\spar{\sarrow{\svar{A}}{\svar{A}}}}{\svar{A}}}}}
\newcommand{\sappidtypf}
  {\sarrow{\svar{A}}{\sarrow{\spar{\sarrow{\svar{A}}{\svar{A}}}}{\svar{A}}}}
\newcommand{\sappidtypff}
  {\sarrow{\spar{\sarrow{\svar{A}}{\svar{A}}}}{\svar{A}}}

\newcommand{\tfapp}{\tabs{\tvar{f}}{\tabs{\tvar{x}}{\sapp{\tvar{f}}{\tvar{x}}}}}
\newcommand{\tfapparg}[1]{\tabs{\tvar{x}}{\sapp{#1}{\tvar{x}}}}
\newcommand{\tfappargs}[2]{\sapp{#1}{#2}}
\newcommand{\tinc}{\tabs{\tvar{x}}{\target{x+1}}}
\newcommand{\tid}{\tabs{\tvar{x}}{\tvar{x}}}

\newcommand{\rexp}[3]{#2 \source{\lesssim}_{#1} #3}
\newcommand{\rval}[3]{#2 \source{\leq}_{#1} #3}
\newcommand{\renv}[2]{\source{\mathcal{G}\llbracket} #1 \source{\rrbracket}_{#2}}
\newcommand{\inrenv}[4]{(#1,#2) \in \renv{#3}{#4}}
\newcommand{\rtypenv}[1]{\source{\mathcal{D}\llbracket} #1 \source{\rrbracket}}
\newcommand{\inrtypenv}[2]{#1 \in \rtypenv{#2}}

\newcommand{\rtexp}[3]{#2 \target{\lesssim}_{#1} #3}
\newcommand{\rtval}[3]{#2 \target{\leq}_{#1} #3}
\newcommand{\rtenv}[2]{\target{\mathcal{G}{\llbracket}} #1 \target{\rrbracket}_{#2}}
\newcommand{\inrtenv}[4]{(#1,#2) \in \rtenv{#3}{#4}}

\newcommand{\setst}[3]{ \{ (#1, #2) ~|~ #3 \}}
\newcommand{\setsingl}[1]{ \{\}}
\newcommand{\setsto}[2]{ \{ #1 ~|~ #2 \}}

\newcommand\defeq{\mathrel{\overset{\makebox[0pt]{\mbox{\normalfont\tiny\sffamily def}}}{=}}}

\newcommand{\sforce}[1]{\source{\mathit{force}}(#1)}
\newcommand{\tforce}[1]{\target{\mathit{force}(#1)}}

\section{Experimental evaluation}
\label{sec:experiments}

We now present an experimental evaluation of \metafstar.
%
%
First, we provide benchmarks comparing our reflective canonicalizer
from \autoref{sec:canon_semiring} to calling the SMT solver directly
without any canonicalization.
%
%
Then, we return to the parsers and serializers from
\autoref{sec:minilowparse} and show
how, for VCs that arise, a domain-specific tactic is much more
tractable than a SMT-only proof.

\subsection{A Reflective Tactic for Partial Canonicalization}
\label{sec:canonicalization}

%
%

\newcommand{\canonfig}{
  \begin{wrapfigure}{r}{0.51\textwidth}
      \small
      \vspace{-8mm}
  \begin{tabular}{|c||c|c|c|c|}
  \hline
          & Rate     & Queries              & Tactic & Total \\
  \hline
   smt1x   & 0.5\%   & 0.216  $\pm$ 0.001   & --     & 2.937 \\
   smt2x   & 2\%     & 0.265  $\pm$ 0.003   & --     & 2.958 \\
   smt3x   & 4\%     & 0.304  $\pm$ 0.004   & --     & 3.022 \\
   smt6x   & 10\%    & 0.401  $\pm$ 0.008   & --     & 3.155 \\
   smt12x  & 12.5\%  & 0.596  $\pm$ 0.031   & --     & 3.321 \\
   smt25x  & 16.5\%  & 1.063  $\pm$ 0.079   & --     & 3.790 \\
   smt50x  & 22\%    & 2.319  $\pm$ 0.230   & --     & 5.030 \\
   smt100x & 24\%    & 5.831  $\pm$ 0.776   & --     & 8.550 \\
   interp  & 100\%   & 0.141  $\pm$ 0.001   & 1.156  & 4.003 \\
   native  & 100\%   & 0.139  $\pm$ 0.001   & 0.212  & 3.071 \\
  \hline
  \end{tabular}
      \vspace{-6mm}
  \end{wrapfigure}
}

In \autoref{sec:canon_semiring}, we have described the
\ls$canon_semiring$ tactic that rewrites semiring expressions into sums
of products.
%
%
%
%
%
%
We find that this tactic significantly improves proof robustness.
The table below compares the success rates and times for the
\ls$poly_multiply$ lemma from \autoref{sec:poly_multiply}.
To test the robustness of each alternative, we run the tests 200 times
while varying the SMT solver's random seed.
The \texttt{smt$i$x} rows represent asking the solver to prove the lemma
without any help from tactics, where $i$ represents the resource limit
(\texttt{rlimit}) multiplier given to the solver.
This \texttt{rlimit} is memory-allocation based and independent of the
particular system or current load.
For the \ls$interp$ and \ls$native$ rows, the
\ls$canon_semiring$ tactic is used, running it using \fstar's KAM normalizer and
as a native plugin respectively---both with an \texttt{rlimit}
of $1$.
For each setup, we display the success rate of verification, the average
(CPU) time taken for the SMT queries (not counting the time for parsing/processing
the theory)
with its standard deviation, and
the average total time (its standard deviation coincides with that of
the queries).
When applicable, the time for tactic execution (which is independent of
the seed) is displayed.
The \texttt{smt} rows show very poor success rates: even when upping the
\texttt{rlimit} to a whopping 100x, over three quarters of the attempts
fail.
\canonfig
Note how the (relative) standard deviation increases with the
\texttt{rlimit}: this is due to successful runs taking rather random
times, and failing ones exhausting their resources in similar times.
The setups using the tactic show a clear increase in
robustness: canonicalizing the assertion causes this proof to always
succeed, even at the default \texttt{rlimit}.
%
%
We recall that the
tactic variants still leave goals for SMT solving, namely, 
the skeleton for the original VC
and the canonicalized equality left by the tactic, easily dischargeable
by the SMT solver through much more well-behaved linear reasoning.
The last column shows that native compilation speeds up this tactic's
execution by about 5x.

\subsection{Combining SMT and Tactics for the Parser Generator}
\label{sec:eval-minilowparse}


In \autoref{sec:minilowparse}, we presented a library of combinators
and a metaprogramming approach to automate the construction of
verified, mutually inverse, low-level parsers and serializers from type
descriptions.
%
%
Beyond generating the code, tactics are used to process and discharge
proof obligations that arise when using the combinators.

We present three strategies for discharging these obligations,
including those of bijectivity that arise when constructing parsers
and serializers for enumerated types. First, we used \fstar's default
strategy to present all of these proofs directly to the SMT
solver. Second, we programmed a $\sim$100 line tactic to discharge
these proofs without relying on the SMT solver at all. Finally, we
used a hybrid approach where a simple, 5-line tactic is used to prune
the context of the proof removing redundant facts before presenting
the resulting goals to the SMT solver.

\begin{wrapfigure}{r}{0.42\textwidth}
    \small
    \vspace{-7mm}
    \begin{tabular}{|c||c|c|c|}
        \hline
                       Size           & SMT only & Tactic only & Hybrid \\
        \hline
                         4            & 178      & 17.3        & 6.6    \\
                         7            & 468      & 38.3        & 9.8    \\
                        10            & 690      & 63.0        & 19.4   \\ \hline
    \end{tabular}
    \vspace{-8mm}
\end{wrapfigure}

The table alongside shows the total time in seconds for verifying
metaprogrammed low-level parsers and serializers for enumerations of
different sizes.
In short, the hybrid approach scales the best; the
tactic-only approach is somewhat slower; while the SMT-only approach
scales poorly and is an order of magnitude slower. Our hybrid approach
is very simple. With some more work, a more sophisticated hybrid
strategy could be more performant still, relying on tactic-based
normalization proofs for fragments of the VC best handled
computationally (where the SMT solver spends most of its time), while
using SMT only for integer arithmetic, congruence closure
etc. However, with \metafstar's ability to manipulate proof contexts
programmatically, our simple context-pruning tactic provides a big
payoff at a small cost.

\section{Related Work}
\label{sec:related}

\ch{Should also cite this very recent F-IDE 2018 paper
\url{https://arxiv.org/abs/1811.10814}}

\ch{Another tool we should probably compare with is KeY, which also
  has a notion of interactive proofs
  \url{https://pdfs.semanticscholar.org/2f5d/6602fbaa97fef0a9df1216451e536c33d617.pdf}.
  With a lot of work going into a usable GUI for proofs:
  \url{https://ieeexplore.ieee.org/document/7582776}.}



Many SMT-based program verifiers~\cite{BarnettCDJL05,
BurdyCCEKLLP05, BarnettDFJLSV05, FlanaganLLNSS13, LeinoN98},
%
%
rely on user hints, in the form of assertions and lemmas, to complete proofs.
This is the predominant style of proving
used in tools like Dafny~\cite{leino10dafny}, Liquid Haskell~\cite{VazouTCSNWJ18}, Why3~\cite{filliatre13why3}, and \fstar
itself~\cite{mumon}.
%
%
%
However, there is a growing trend to augment this style of semi-automated proof with interactive proofs.
For example, systems like Why3~\cite{filliatre13why3} allow
VCs to be discharged
using ITPs such as Coq, Isabelle/HOL, and PVS, but this requires
an additional embedding of VCs into the logic of the ITP in question.
In recent concurrent work, support for {\em effectful} reflection
proofs was added to Why3~\cite{why3reflection}, and it would be
interesting to investigate if this could also be done in \metafstar.
\citet{GrovT16} present Tacny, a tactic framework for Dafny,
which is, however, limited in that it only transforms source code,
with the program verifier unchanged.
%
In contrast, \metafstar{} combines the benefits of an SMT-based
program verifier and those of tactic proofs within a single language.

Moving away from SMT-based verifiers, ITPs have long
relied on separate languages for proof scripting,
starting with Edinburgh
LCF~\cite{gordon-edinburghLCF:1979} and ML,
%
and continuing with HOL, Isabelle and Coq, which are either extensible
via ML, or have dedicated tactic languages~\cite{isar, Delahaye00,
  AnandBCST18, StampoulisS10}.
%
%
%
%
%
%
\metafstar builds instead on a recent idea in the space of dependently
typed ITPs \cite{ZilianiDKNV15, Mtac2, ChristiansenB16,
  EbnerURAM17} of reusing the 
object-language as the meta-language.
This idea first appeared in Mtac, a Coq-based tactics framework
for Coq%
~\cite{ZilianiDKNV15, Mtac2},
%
and has many generic benefits including reusing the standard library,
IDE support, and type checker of the proof assistant.
Mtac can additionally check the partial
correctness of tactics, 
which is also sometimes possible in \metafstar but still rather limited
(\autoref{sec:specifying-metaprograms}).
%
%
\metafstar's design is instead more closely inspired by the
metaprogramming frameworks of Idris~\cite{ChristiansenB16} and
Lean~\cite{EbnerURAM17}, which provide a deep embedding of terms that
metaprograms can inspect and construct at will without dependent types
getting in the way.
However, \fstar's effects, its weakest precondition calculus, and its
use of SMT solvers distinguish \metafstar from these other frameworks,
presenting both challenges and opportunities, as discussed in this
paper.


Some SMT solvers also include tactic engines~\cite{smttactics},
which allow to process queries in custom ways.
However, using SMT tactics from a program verifier is not very
practical.
To do so effectively, users must become familiar not only with the
solver's language and tactic engine, but also with the translation from
the program verifier to the solver.
Instead, in \metafstar, everything happens within a single language.
Also, to our knowledge, these tactics are usually coarsely-grained,
and we do not expect them to enable developments such as
\autoref{sec:seplogic}.
Plus, SMT tactics do not enable metaprogramming.


Finally, ITPs are seeing increasing use of ``hammers'' such as
Sledgehammer~\cite{BlanchetteBP13, PaulsonB10, Blanchette013} in
Isabelle/HOL, and similar tools for HOL Light and
HOL4~\cite{KaliszykU14}, and Mizar~\cite{KaliszykU15a}, to interface
with ATPs.
This technique is similar to \metafstar, which, given its support for
a dependently typed logic is especially related to a recent hammer for
Coq~\cite{CzajkaK17}. Unlike these hammers, \metafstar does not aim to
reconstruct SMT proofs, gaining efficiency at the cost of trusting the
SMT solver. Further, whereas hammers run in the background, lightening
the load on a user otherwise tasked with completing the entire
proof, \metafstar relies more heavily on the SMT solver as an end-game
tactic in nearly all proofs.




\section{Conclusions}
\label{sec:conclusions}

A key challenge in program verification is to balance
automation and expressiveness. Whereas tactic-based ITPs
support highly expressive logics, the tactic author is responsible for all
the automation.
Conversely, SMT-based program verifiers
provide good, scalable automation for comparatively weaker logics, but offer
little recourse when verification fails. A design that allows picking the
right tool, at the granularity of each verification sub-task, is a
worthy area of research.
\metafstar presents a new point in this space: by using hand-written
tactics alongside SMT-automation, we have written
proofs that were previously impractical in \fstar, and
(to the best of our
knowledge) in other SMT-based program verifiers.


\iffull
This flurry of new use-cases is backed by an efficient native
compilation scheme.
%
Natively evaluating metaprograms allows to dynamically extend
\fstar with type-safe, semantically correct (by the safety of
\ls$TAC$), user-defined custom behavior without a performance
penalty---witnessing the close interoperability between \fstar and its
metalanguage.
\fi

\newcommand\grantsponsor[3]{#2}
\newcommand\grantnum[2]{#1-#2}

\ifcamera
\paragraph{Acknowledgements}
  We thank Leonardo de Moura
  and the Project Everest team for many useful discussions.
  The work of
    Guido Mart\'inez,
    Nick Giannarakis,
    Monal Narasimhamurthy,
  and Zoe Paraskevopoulou
  was done, in part, while interning at Microsoft Research.
  Cl\'ement Pit-Claudel's work was in part done
  during an internship at Inria Paris.
  The work of Danel Ahman, Victor Dumitrescu, and C\u{a}t\u{a}lin Hri\c{t}cu
  is supported by the MSR-Inria Joint Centre and the
  \grantsponsor{1}{European Research Council}{https://erc.europa.eu/}
  under ERC Starting Grant SECOMP (\grantnum{1}{715753}).
\fi


\iffull
\newpage
\appendix
\begin{figure}
\begin{lstlisting}
Lemma poly_multiply
  (n p r h r0 r1 h0 h1 h2 s1 d0 d1 d2 hh : Z)
  (H: p > 0 /\ r1 >= 0 /\ n > 0 /\ 4 * (n * n) = p + 5 /\ r = r1 * n + r0 /\
            h = h2 * (n * n) + h1 * n + h0 /\ s1 = r1 + (r1 / 4) /\ r1 mod 4 = 0 /\
            d0 = h0 * r0 + h1 * s1 /\ d1 = h0 * r1 + h1 * r0 + h2 * s1 /\
            d2 = h2 * r0 /\ hh = d2 * (n * n) + d1 * n + d0)
  : ((h * r) mod p = hh mod p).
Proof.
  repeat match goal with H : ?A /\ ?B |- _ => destruct H end.
  subst.
  match goal with |- ?A mod p = ?B mod p => change (eqm p A B) end.
  match goal with K : r1 mod 4 = 0 |- _ =>
   apply Z_div_exact_full_2 in K; try omega;
   revert K;
   generalize (r1 / 4);
   intro u;
   intros; subst
  end.
  pose (b := (h2 * n + h1) * u).
  generalize (Zopp_eqm p).
  generalize (Zplus_eqm p).
  generalize (Zminus_eqm p).
  generalize (Zmult_eqm p).
  generalize (eqm_setoid p).
  set (e := eqm p).
  intros.
  match goal with |- e _ ?R =>
    generalize (modulo_addition_lemma R p b)
  end.
  fold e. intro K. setoid_rewrite <- K. clear K.
  assert (p = 4 * (n * n) - 5) as J by omega.
  replace (b * p) with (b * (4 * (n * n) - 5)) by congruence.
  unfold b.
  apply eq_implies_eqm.
  ring.
Qed.
\end{lstlisting}
\caption{Coq proof of \ls$poly_multiply$}
\label{fig:coqpolymultiply}
\end{figure}

\section{Metaprogramming verified parsers and serializers}
\label{sec:parsers-detail}

We consider parser and serializer specifications as ghost code that
``operates'' on finite sequences of bytes of unbounded lengths.
For a given type $t$, a serializer for $t$ marshals any data of type
$t$ into a finite sequence of bytes; and a parser for $t$ reads from a
sequence of bytes, checks whether it corresponds to a valid data of
type $t$, and if so, returns that parsed data and the number of bytes
consumed.

The serializer specification always succeeds; the parser specification
succeeds if and only if the input data is valid with respect to the
parser. On those specifications, we aim to prove that the parser and
the serializer are partial inverses of each other. We encode these
requirements as refinements on the types of parser and serializer
specifications:
\begin{lstlisting}
type byte = FStar.UInt8.t
type parser t = (input: seq byte) -> GTot (option (t * (x: nat { x < length input } )))
let serialize_then_parse_eq_id #t (p: parser t) (s: (t -> GTot (seq byte))) =
  (forall (x: t) . let y = s x in p y == Some (x, length y))
let parse_then_serialize_eq_id #t (p: parser t) (s: (t -> GTot (seq byte))) =
  (forall (x: seq byte) . match p x with | Some (y, len) -> s y == slice x 0 len | _ -> True) })
type serializer #t (p: parser t) = (s: (t -> GTot (seq byte))
  { serialize_then_parse_eq_id p s /\ parse_then_serialize_eq_id p s })
\end{lstlisting}

Whereas those specifications are ghost code, we would like to generate
implementations that can be extracted to C. To this end, we generate
\emph{stateful} implementations written in the \lowstar set of \fstar~\cite{lowstar},
operating on buffers, which are \lowstar mutable data structures
representing C arrays. There, instead of a sequence of bytes, the
parser implementation is given an input buffer and its length; and the
serializer implementation is given an output buffer (along with its
length) onto which it is to serialize the data. This means that, given
a piece of data to serialize and a destination buffer, a serializer
implementation can succeed only if the serialized data fits into the
destination buffer; if so, then the serializer will return the number
of bytes written.

We bake the correctness of the implementations with respect to their
specifications at the level of their types:

\begin{lstlisting}
type buffer8 = LowStar.Buffer.buffer FStar.UInt8.t
type u32 = FStar.UInt32.t

let parser32_postcond #t (p: parser t) (input: buffer8) (l: u32 {l == len input})
  (h: mem) (res: option (t * u32)) (h': mem) =
  modifies loc_none h h' /\ (* memory safety *)
  match p (as_seq h input), res with (* functional correctness *)
  | Some (x, ln), Some (x', ln') -> x' == x /\ FStar.UInt32.v ln' == ln
  | None, None -> True   | _ -> False

type parser32 #t (p: parser t) = (input: buffer8) -> (l: u32 { l == len input } )
                                              -> Stack (option (t * u32))
  (requires (fun h -> live h input)) (ensures (parser32_postcond p input))

let serializer32_postcond #t (p: parser t) (s: serializer p) (output: buffer8)
  (l: u32 {l == len output}) (x: t) (h: mem) (res: option u32) (h': mem) =
  live h output /\ live h' output /\ modifies (loc_buffer b) h h' /\ (* memory safety *)
  let ln = length (s x) in
  match res with (* functional correctness *)
  | Some ln' ->
    FStar.UInt32.v ln' == ln /\ ln <= FStar.UInt32.v l /\
    let b' = sub b 0ul ln' in modifies (loc_buffer b') h h' /\ as_seq h' b' == s x
  | None -> ln > FStar.UInt32.v l

type serializer32 #t (#p: parser t) (s: serializer p) =
  (output: buffer8) -> (l: u32 { l == len output } ) -> (x: t) -> Stack (option u32)
  (requires (fun h -> live h output)) (ensures (serializer32_postcond s output l x))
\end{lstlisting}

\tahina{TODO: gen\_parser can be implemented using typeclasses, triggering gen\_enum\_parser.}

\paragraph{\bf Generating the specification from a \fstar datatype}

Instead of mandating the user to write the parser and serializer
specification themselves, we write a \metafstar
tactic, \verb+gen_specs+, to generate both
the parser and the serializer specifications directly from the \fstar
type $t$ given by the user. These tactics operate by syntax inspection
on $t$ itself, and generates the parser and serializer specifications
using our combinators.

Contrary to Coq's Ltac, \metafstar tactics can also inspect
definitions of inductive types. From there, we made \verb+gen_specs+
also generate a parser and serializer specification out of an
enumeration type defined as a \fstar inductive type whose constructors
have no arguments: for such a type with $n \leq 256$ constructors, the
corresponding parser will associate a unique 8-bit integer from $0$ to
$n-1$ to each constructor. For instance, if $t$ is defined as a \fstar
enumeration type:

As before, by virtue of the correctness of serializers being baked in
their type, the correctness of the serializer generated
by \verb+gen_specs+ does not appear as such as a verification
condition when \fstar is type-checking the generated term. In the
example above, the only proof obligations generated are those to
correctly type-check the rewriting functions themselves passed
to \verb+serialize_synth+ (e.g. enum constructors are rewritten to
integers less than 4), and the precondition of \verb+serialize_synth+
(i.e. the fact that the two rewriting functions are inverse of each
other.) These can be discharged by the SMT solver, but they can also
be discharged automatically by our tactics directly by case analysis
on the inductive type, or by bounded integer enumeration.

\paragraph{\bf Binders}

For instance, we specified \verb+seq_p+, a combinator for dependent
parsing:\protect{\footnote{This parsing combinator is based on yet
unpublished work by Tej Chajed. Building serializers for parsers
derived from \ls$seq_p$ is explicitly out of the scope of this
paper.}} \tahina{I mentioned Tej in the footnote. This part of the
footnote should be removed if Tej becomes a coauthor of the paper.}

\begin{lstlisting}
let seq_p #t1 (p1: parser t1) #t2 (p2: (t1 -> Tot (parser t2))) : Tot (parser t2) =
  fun input -> match p1 input with | None -> None | Some (x1, cons1) ->
    let input2 = slice input cons1 (length input) in
    match p2 x1 input2 with
    | None -> None | Some (x2, cons2) -> Some (x2, cons1 + cons2)
\end{lstlisting}

Then, we implemented it in \lowstar (\verb+seq_pi+), and
manually proved it correct. From there, the user can write a simple
parser specification for a parser that parses an integer encoded
in one or two bytes depending on its value (where \verb+parse_ret x+
is a parser that returns \verb+x+ without reading its byte input:)

\begin{lstlisting}
let example : parser_spec FStar.Int16.t =
  seq_p parse_u8
           (fun lo -> if lo <^ 128uy then parse_ret (cast_u8_to_i16 lo)
                   else parse_synth parse_u8 (fun hi -> cast_u8_to_i16 (lo %^ 128uy)
                                                 +^ cast_u8_to_i16 hi))
\end{lstlisting}
Then, if the user writes:
\begin{lstlisting}
let example_impl : parser_impl example = _ by gen_parser_impl
\end{lstlisting}
then \verb+gen_parser32+ inspects the shape of the \emph{goal},
which is \verb+parser32 example+, and so generates the following \lowstar implementation:
\begin{lstlisting}
seq_pi parse32_u8 (fun lo ->
  parse32_ifthenelse (lo <^ 128uy) (fun _ -> parse32_ret (cast_u8_to_i16 lo)) (fun _ -> parse32_synth parse32_u8 (fun hi -> cast_u8_to_i16 (lo %^ 128uy) +^ cast_u8_to_i16 hi)))
\end{lstlisting}

from where KreMLin \cite{lowstar} then inlines the corresponding implementation combinators to produce C code.

First, \verb+gen_parser32+ inspects the shape of the goal to determine
the type of the data to be parsed and the parser specification; then,
it calls a recursive tactic \verb+gen_parser32'+ that inspects the
shape of the parser specification and actually builds the
implementation.

\begin{lstlisting}
module T = FStar.Tactics
let rec gen_parser32' (env: T.env) (t: T.term) (p: T.term) : T.Tac T.term =
  let (hd, tl) = app_head_tail p in
  if hd `T.term_eq` (`(parse_ret)) then T.mk_app (`(parse32_ret)) tl else
  if hd `T.term_eq` (`(parse_u8)) then (`(parse32_u8)) else
  if hd `T.term_eq` (`(seq_p)) then match tl with
  | [(t, _); (p, _); (t', _); (p', _)] ->
    begin match T.inspect p' with
    | T.Tv_Abs bx body ->
      let p32 = gen_parser32' env k t p in
      let env' = T.push_binder env bx in
      let body' = gen_parser32' env' k' t' body in
      let p32' = T.pack (T.Tv_Abs bx body') in
      (`(seq_pi #(`#t) #(`#p) (`#p32) (`#t') (`#p') (`#p32')))
    | _ -> ...
    end
  | _ -> ...
  else
  ...

let gen_parser32 () : T.Tac unit =
  let (hd, tl) = app_head_tail (T.cur_goal ()) in
  if hd `T.term_eq` (`parser32)
  then match tl with
  | [(t, _); (p, _)] ->
    let env = T.cur_env () in
    let p32 = gen_parser32' env t p in
    T.exact_guard p32;
  | _ -> ...
\end{lstlisting}

When \verb+gen_parser32+ (resp. \verb+gen_serializer32+) builds the
parser (resp. serializer) implementation, it may generate verification
conditions due to specific preconditions required by certain
combinators (such as the precondition of the \verb+serialize_synth+
rewriting serializer combinator: the rewriting functions being inverse
of each other). Contrary to Coq, there are no proof objects, so those
preconditions need to be proven again even if they had been proven
earlier by the user at the specification level. Nevertheless, those
preconditions can be automatically solved by our tactics, or they can
directly be sent to the SMT solver.

However, by virtue of the correctness of the implementation
combinators being baked in their type, typechecking the resulting term
generated by \verb+gen_parser32+ triggers no verification condition
other than those preconditions required by the specific combinators,
and unification constraints, the latter solved by reflexivity. In
particular, the memory safety and correctness of the implementations
generated by \verb+gen_parser32+ do not appear as such as verification
conditions when \fstar is type-checking the terms generated
by \verb+gen_parser32+. \tahina{Is it already important to mention
here the kinds of proof obligations generated by the tactics, and that
they can be solved either by SMT or tactics? Or should we move to
section 5?}

\subsection{Verified meta-programming of program transformations}

\ls$compile_bind$ expects, in \ls$bind f1 f2$, that \ls$f2$ be an
abstraction, and so needs to manipulate binders and recursively
compile \ls$f1$ and the body of \ls$f2$. This will, in the latter
case, make some bound variables appear free in terms being compiled.

If the head term $t_0$ is a free variable not corresponding to any
combinator, then \ls$compile_fvar$ unfolds it, compiles the unfolded
term, and inserts a coercion.\footnote{We defined an explicit coercion
combinator, \ls$coerce_sz$, such that \ls$coerce_sz f1 f1_sz f2$ is
well-typed and is a \ls$m_sz f2$ as soon as \ls$f1 () == f2 ()$ as $m$
specifications.} In most cases, a free variable to be unfolded will
correspond to a function definition, so that unfolding will yield a
$\beta$-redex, which \ls$compile_fvar$ reduces by calling a
normalization tactic that we provide as part of \metafstar. This needs
an \emph{environment}
\ls{(e: env)} tracking the bound variables encountered by nested calls
to \ls$compile_bind$.

We show below the implementation of \ls$compile_bind$, which is
representative of most features of \metafstar that we use for
our \ls$compile$ tactic.

\begin{lstlisting}
 (* extract the head symbol and the arguments of an application *)
let rec app_head_rev_tail (t: term) : Tac (term * list argv) = 
  match inspect t with | Tv_App u v -> let (x, l) = app_head_rev_tail u in (x, v :: l) | _ -> (t, [])
let app_head_tail (t: term) : Tac (term * list argv) = let (x, l) = app_head_rev_tail t in (x, rev l)

let compile_bind (e: env) (ty: term) (t: term) (compile: env -> term -> term -> Tac term) : Tac term =
  let (t_0, ar) = app_head_tail t in
  guard (t_0 = quote bind); (* fail if false *)
  match ar with
  | [ ($\tau_1$, _); ($\tau_2$, _); (f1, _); (f2, _); ] ->
    begin match inspect f2 with
    | Tv_Abs v f2_body ->
      let f1' = compile e $\tau_1$ f1 in
      let e' = push_binder e v in
      let f2_body' = compile e' $\tau_2$ f2_body in
      let f2' = pack (Tv_Abs v f2_body') in
      (`(bind_sz #(`#tau_1) #(`#tau_2) (`#f1) (`#f1') (`#f2) (`#f2')))
    | _ -> fail "compile_bind: Not an abstraction"
    end
  | _ -> fail "compile_bind: 4 arguments expected"
\end{lstlisting}

\paragraph{\bf Putting everything together}
Finally, the following wrapper combinator $\Phi$ puts everything
together, taking a specification $f$ and its size and buffer
implementations, and yielding a stateful function, fitting in
the \lowstar subset of \fstar amenable to extraction to C, that first
computes the expected size of the output, then allocates one single
buffer using the C \ls$malloc$, then uses it to write the output of
$f$:
\begin{lstlisting}
let $\Phi$ $\tau$ (f: m $\tau$) (f_sz: m_sz f) (f_st: m_st f) : ST (option ($\tau$ * buffer U8.t)) (requires (fun _ -> True))
  (ensures (fun h res h' -> let (r, log) = f () in
    if length log > $2^{31} - 1$ then res == None else
    Some? res /\ (let (Some (r', b)) = res in r == b /\ fresh b h h' /\  as_seq h b == log)
  )) = match f_sz with | None -> None | Some (_, sz') ->
       let b = rcreate_mm root 42uy sz in (* then allocate a fresh buffer *)
       let (r, _) = f_st b in (* then write the output into the buffer *)
       Some (r, b)
\end{lstlisting}

\paragraph{\bf Related Work}
Our approach of generating verified low-level formatters is related to
Amin and Rompf's~\cite{Amin:2017:LAW:3009837.3009867} work on
LMS-Verify. They use metaprogramming facilities in Scala to generate C
code together with proof annotations to be checked by
Frama-C~\citep{CuoqKKPSY12}, an SMT-based C verifier. In contrast, we
generate correct-by-construction imperative \lowstar code, which is
verified by \fstar{} before being translated to C by the KreMLin
compiler.  LMS-Verify has also been used to generate efficient and
safe HTTP parsers, a larger-scale effort than the network message
formatters than we have currently done.


\newpage

\section{Proof of \autoref{thm:split}}
\label{app:correctness}
This proof is based on the formal definition of \emf, a recent formalization
of an \fstar{} subset~\cite{dm4free}.
We use the same notation and rule names from there, but the proof is
nevertheless quite direct with just minimum familiarity with \emf.
We use $E$ to represent \fstar contexts:

\[
\begin{array}{rclcl}
    E &=&    \cdot  & & \\
      & \mid & t E &                 \mid & E t \\
      & \mid & \lambda(x:t). E &     \mid & \lambda(x:E). t \\
      & \mid & (x:t) \rightarrow E & \mid & (x:E) \rightarrow t \\
      & \mid & x:t\{E\} &            \mid & x:E\{t\} \\
      & ~ & \ldots
\end{array}
\]

Contexts present a set of bound variables to their hole. We denote
those variables by $\gamma(E)$:

\[
\begin{array}{lcl}
    \gamma(\cdot) &=& \emptyset \\
    \gamma(\lambda(x:t). E) &=& \{x\} \cup \gamma(E) \\
    \gamma(\lambda(x:E). t) &=& \gamma(E) \\
    \ldots
\end{array}
\]
    % ...

We say a context $E$ has type $t_1 \Rightarrow t_2$ for $\Gamma$ (noted
by $\Gamma \vdash E : t_1 \Rightarrow t_2$) when for all $e$ such that
$\Gamma, \gamma(E) \vdash e : t_1$ we also have $\Gamma \vdash E[e] : t_2$.

\subsection{Soundness of splitting the proof obligation}

The proof follows mainly from the following lemma:

\begin{lemma}
    \[
    \infer{\Gamma \vDash E[t_1] = E[t_2]}
          {\Gamma, \gamma(E) \vDash t_1 = t_2}
    \]
\end{lemma}
\begin{proof}
    The proof follows from applying functional extensionality and using
    a carefully crafted abstraction.

    Let $S = \lambda f. E[f (\gamma(E))]$ (where $f (\gamma(E))$
    represents $f$ applied to every variable in $\gamma(E)$,
    sequentially). Also, $\lambda \gamma(E). t$ represents
    abstracting $E$ over the variables in $\gamma(E)$, and
    similarly for $\forall \gamma(E). t$.

    Then,
    \[
        \infer[\mbox{Reduction}]{\Gamma \vDash E[t_1] = E[t_2]}
        {\infer[\mbox{V-EqP}]{\Gamma \vDash S (\lambda \gamma(E). t_1) = S (\lambda \gamma(E). t_2)}
        {\infer[\mbox{V-Ext}]{\Gamma \vDash (\lambda \gamma(E). t_1) = (\lambda \gamma(E). t_2)}
        {\infer[\mbox{V-$\forall$i}]{\Gamma \vDash \forall \gamma(E). t_1 = t_2}
                             {\Gamma, \gamma(E) \vDash t_1 = t_2}}}}
    \]

    Note that the particular set of contexts E does not influence
    the proof.

\end{proof}

Having such lemma, the theorem is proven as follows:
\[
    \infer{\Gamma \vDash E[\phi]}
         {\Gamma \vDash E[\top]
         & \infer[\mbox{Lemma}]{\Gamma \vDash E[\phi] = E[\top]}{
         \infer[\mbox{PropExt}]{\Gamma, \gamma(E) \vDash \phi = \top}{
         \infer[\mbox{Trivial}]{\Gamma, \gamma(E) \vDash \phi \iff \top}
                               {\Gamma, \gamma(E) \vDash \phi}}}
         }
\]

\subsection{Partial completeness of splitting the proof obligation}

While the previous theorem allows to soundly split any
subformulae within a VC, we have seen that some of them
constraint the system.
Here we prove that for a particular set of contexts, splitting
does not make the judgments stronger, and via Theorem 1, then they are
equivalent.

We define a particular shape of ``positive'' contexts $P$.
We don't attempt to follow all of EMF*'s syntax, just the parts
we commonly see as VCs in practice.

\[
\begin{array}{rclcl}
    P &=&    \cdot  & & \\
       & \mid & \phi \land P            &\mid& P \land \phi \\
       & \mid & \forall (x:t). P \\
\end{array}
\]

In \emf, an implication $a \Rightarrow b$ is just sugar for
$\forall (\_:a). b$, so they are considered as well.


\begin{theorem}
    Take $\Gamma$, $E$, and $\phi$ such that:
    (1) $\phi$ is squashed,
    (2) $\Gamma \vdash E : prop \rightarrow prop$
    and (3) $\Gamma \vdash \phi : prop$.
    Then the following holds

    \[
        \infer{\Gamma, \gamma(P) \vDash \phi \qquad \Gamma \vDash P[\top]}
              {\Gamma \vDash P[\phi]}
    \]
\end{theorem}
\begin{proof}
    By induction on the structure of $P$, keeping $\Gamma$ and $\phi$
    universally quantified.
    \begin{itemize}
        \item For $P = \cdot$, trivial
        \item For $P = \psi \land P'$, our hypothesis is
              $\Gamma \vDash \psi \land P'[\phi]$.
              Hence,
              $\Gamma \vDash \psi$ and
              $\Gamma \vDash P'[\phi]$.
              By the IH, we get
              $\Gamma, \gamma(P') \vDash \phi$
              (note that $\gamma(P) = \gamma(P')$)
              and
              $\Gamma \vDash P'[\top]$.
              By weakening and conjunction, 
              we can prove $\Gamma, \gamma(P') \vDash \psi \land \phi$
              and $\Gamma \vDash \psi \land P'[\top]$, and conclude.

        \item For $P = P' \land \psi$, the reasoning is analogous.

        \item For $P = \forall (x:t). P'$, our hypothesis is
              $\Gamma \vDash \forall (x:t). P'[\phi]$. Hence,
              we get $\Gamma, x:t \vDash P'[\phi]$ by
              eliminating the quantifier.
              From the IH, we get
              $\Gamma, x:t, \gamma(P') \vDash \phi$
              and
              $\Gamma, x:t \vDash E[\top]$.
              From this last one, we can use V-$\forall$i to
              obtain $\Gamma \vDash \forall (x:t). E[\top]$, and conclude
              (noting that $\gamma(P) = x:t, \gamma(P')$).
    \end{itemize}
\end{proof}

\newpage

\section{Modelling native plugins}
\label{app:native}

This appendix contains the formal definitions of our multi-language
interoperability model. We formalize the semantics of source and target language
as well as the translation between the two. We convey the essential ideas 
in the text below, and present full details in 
Figures \ref{fig:sourcelanglong}--\ref{fig:translationlong} that follow the discussion below.

\subsection*{Part 1: Modeling Native Plugins with Simple Types}

Reflecting our requirement that plugins are only supported at
ML-typeable interfaces, we start with a
source language that is an intrinsically typed (Church
style), standard, simply typed lambda calculus with pairs and
sums. Later in the section, we
will add ML-style rank 1 polymorphism to it.
Conversely, reflecting
the type-less native representation of compiled OCaml code, our target
language is the untyped lamdba calculus.
For clarity, we markup the the syntax using colors, using
$\source{blue}$ for the source language and $\target{red}$ for the target.
Aside from the standard forms, we
have two new expression forms to account for the ``alien''
expressions of one language appearing in the other:
$\talien{\sexp}{\stau}{}$, which \emph{embeds} a
$\stau$-typed source language expression into the target language, and 
$\salien{\texp}{\stau}{}$,
which \emph{unembeds} a target language expression to the source language at
type $\stau$.



We explain the reductions through a small example. Consider a source language term
$\sapp{\source{(}\sabs{\svar{x}}{\mathbb{\source{Z}}}{\svar{x}}\source{)}}{\source{0}}$, where we want to compile the
identity function to native code and then perform the application. The
first step is to translate the function to the target, by systematically
erasing the types, and unembed it in the source. In our example, the
source language term then becomes
$\sapp{\salien{\tabs{\tvar{x}}{\tvar{x}}}{\mathbb{\source{Z}} \source{->} \mathbb{\source{Z}}}{}}{\source{0}}$.





%

\paragraph{\bf Source semantics and unembedding (\autoref{fig:sourceseman})}
At a beta-reduction step (e.g. rule {\sc{S-App}} below), the source
semantics uses a meta-function
$\sforce{\source{v}}$ to examine the head of the application $\svar{v}$.
In addition to the usual values, the term
$\salien{\tval}{\stau}{\relax}$ is a source value \emph{iff}
$\tval$ is a target value. If the head is such a
value, $\sforce{}$ invokes the unembedding coercion
$\scoerce\stau\relax\tval$ to coerce it to a suitable source
value. For structural types, this amounts to lazily coercing the head
constructor of the term (see below). For function types, a source-level closure is
allocated, which first \emph{embeds} its argument in the target,
reduces a target application (using rule {\sc{S-Alien}}), and
then \emph{unembeds} the result back in the source.
%
%
%
%
%
%
\begin{figure}
\begin{mathpar}
  \inferrule*[right=S-Alien]
             {\texp \tevals \tval}
             {\salien{\texp}{\stau}{} \seval \salien{\tval}{\stau}{}}
  %
 %
  \and
  \inferrule*[right=S-App]
             {\sforce{\sval}=\sabs{\svar{x}}{\stau}{\sexp}}
             {\sapp{\sval}{\sexpii} \seval \sexp \ssubst{\svar{x}}{\sexpii}}
 \and
  \inferrule*{}{
    \sforce{\sval} =
    \begin{cases}
      \scoerce{\stau}{\relax}{\tval}
      & \text{if }\sval = \salien{\tval}{\stau}{\relax} \cr
        \sval
      & \text{otherwise}
    \end{cases}
  }
\end{mathpar}
\vspace{0.05cm}
\[
\begin{array}{r l l}
  \scoerce{\stunit}{}{\tunit} & = & \sunit

\\[0.3em]

  \scoerce{\sprod{\staui}{\stauii}}{\relax}{\tpair{\tvali}{\tvalii}} & = &
  \spair{\salien{\tvali}{\staui}{\relax}}{\salien{\tvalii}{\stauii}{\relax}}

\\[0.3em]

  \scoerce{\sarrow{\staui}{\stauii}}{\relax}{\texp} & = &
  \sabs{\svar{x}}{\staui}{\salien{\tapp{\texp}{\talien{\svar{x}}\staui\relax}}{\stauii}{\relax}}
\end{array}
\]
    \caption{Source semantics}
    \label{fig:sourceseman}
\end{figure}

In our example, reduction proceeds by rule {\sc{S-App}} by first
applying the unembedding coercion
$\scoerce{\source{\mathbb{Z} -> \mathbb{Z}}}{}{}$ to $\tabs{\tvar{x}}{\tvar{x}}$ to get
$\sabs{\svar{x}}{\mathbb{\source{Z}}}{\salien{\tapp{\target{(}\tabs{\tvar{x}}{\tvar{x}}\target{)}}{\talien{\svar{x}}{\mathbb{\source{Z}}}{}}}{\mathbb{\source{Z}}}{}}$,
and then applying the substitution to get
$\salien{\tapp{\target{(}\tabs{\tvar{x}}{\tvar{x}}\target{)}}{\talien{\source{0}}{\mathbb{\source{Z}}}{}}}{\mathbb{\source{Z}}}{}$.

%
%

\paragraph{\bf Target semantics and embeddings (\autoref{fig:targetseman})}
The target semantics is standard
 CBV.
 The only new rule is \textsc{T-Alien}, which first reduces the foreign
 source term using  the source semantics, and then proceeds
 to \emph{embed} the resulting value into the target
 language, using the embedding coercion $\tcoerce\stau\relax\sval$,
 as shown below. In our implementation, this rule is realized via a callback
 to the \fstar's normalizer.

 \begin{figure}
\begin{mathpar}
\inferrule*[right=T-Alien]
            {\sexp \sevals \sval}
            {\talien{\sexp}{\stau}{} \teval \tcoerce{\stau}{}{\sval}}
\and
\inferrule*[right=T-App]
            {\quad}
            {\tapp{\target{(}\tabs{\tvar{x}}{\texp}\target{)}}{\tval} \teval \texp\tsubst{\tvar{x}}{\tval}}
\end{mathpar}
\[
\begin{array}{r l l}
  \tcoerce{\sarrow{\staui}{\stauii}}{\relax}\sexp  & = & 
  \tabs{\tvar{x}}{\talien{\sapp\sexp{\salien{\tvar{x}}{\staui}{\relax}}}{\stauii}{\relax}}
  \\[0.3em]
  \tcoerce{\sprod{\staui}{\stauii}}{\relax}{\spair{\sexpi}{\sexpii}}  & = & 
  \tpair{\talien{\sexpi}{\staui}{\relax}}{\talien{\sexpii}{\stauii}{\relax}}
  \\[0.3em] 
  \tcoerce{\staup}{\relax}{\salien{\tval}{\stau}{\relax}}  & = &  \tval
  \text{ if } \stau \equiv \staup
\end{array}
\]
     \caption{Target semantics}
     \label{fig:targetseman}
 \end{figure}

The last case in the definition of the coercion cancels
 superfluous embeddings, in case the term inside is an
 unembedding. Continuing with our example,
$\salien{\tapp{\target{(}\tabs{\tvar{x}}{\tvar{x}}\target{)}}{\talien{\source{0}}{\mathbb{\source{Z}}}{}}}{\mathbb{\source{Z}}}{}$
 reduces using {\sc{S-Alien}}, which reduces the application in its
 premise by first reducing
 $\talien{\source{0}}{\mathbb{\source{Z}}}{}$
 to
$\target{0}$ (using rule {\sc{T-Alien}}, with
 $\tcoerce{\mathbb{Z}}{\relax}{\source{0}} = \target{0}$), and
then using rule {\sc{T-App}} to get
 $\target{0}$. The rule {\sc{S-Alien}} then returns
 $\salien{\target{0}}{\mathbb{\source{Z}}}{}$
 to the \fstar normalizer, which applies the unembedding coercion to
 get $\source{0}$.
 
\iffull

\paragraph{\bf Translation from source to target}
Fig.~\ref{fig:translation} shows a few
selected cases in the translation of source to target terms.
The key rule is \textsc{Abs}, which translates a source
abstraction to a target abstraction. Rather than maintaining an explicit mapping of
source variables to target variables, the translation \emph{embeds} the target
variable $\tvar x$
in the source body of the abstraction. Later on, the \textsc{Box} rule finds the
embedding, and removes it, leaving only the intended variable in the translated
target program. As such, any closed source term is translated to a closed target
term, without any embeddings or unembeddings in it.

Rule \textsc{Var} is therefore only triggered for open variables that have not
been unembedded by \textsc{Abs} already. The formalization accounts
for this, and remains valid in the presence of open
terms. From an implementation perspective, however, we cannot in general compile
open terms to OCaml. Therefore, we restrict ourselves to compiling top-level
definitions, in which the only open variables are other top-level definitions.
Such top-level definitions have a name, which we know how to compile to OCaml.
Because of this restriction, plugin extraction works at the level of top-level
definitions, which are the only AST nodes that can be marked with the
\ls$plugin$ attribute.

\begin{figure}[h]
  \centering
\begin{mathpar}
  \inferrule*[right=Var]
             {\\}
             {\svar{x} \trans \talien{\svar{x}}{\senv\spar{\svar{x}}}{\cdot}}
  \and
  \inferrule*[right=Box]
             {\\}
             {\salien{\texp}{\stau}{\sdelta} \trans \texp}
  \and
  \inferrule*[right=Pair]
             {\source{e_i} \trans \target{e_i}}
             { \spair{\source{e_1}}{\source{e_2}} \trans \tpair{\target{e_1}}{\target{e_2}}}
  \and
  \inferrule*[right=Abs]
             {\sexp\ssubst{\svar{x}}{\salien{\tvar{x}}{\source{\tau}}{\cdot}} \trans \texp}
             {\sabs{\svar{x}}{\stau}{\sexp} \trans \tabs{\tvar{x}}{\texp}}
  \and
  \inferrule*[right=App]
             {\source{e_i} \trans \target{e_i}}
             {\sapp{\source{e_1}}{\source{e_2}} \trans \tapp{\target{e_1}}{\target{e_2}}}
           \end{mathpar}
           \caption{Translation from source to target}
           \label{fig:translation}
\end{figure}
\fi

\subsection*{Part 2: Adding ML-style Polymorphism} 
\label{sec:parametricity}


The key observation in adding ML-style polymorphism is as
follows. Since we only consider target terms
compiled from well-typed source terms, by virtue of
parametricity, a compiled polymorphic function must treat its
argument abstractly. As a result, the normalizers can pass such
arguments \emph{as-is} without applying any embedding coercions---passing
an open term is simply a subcase without further
difficulty. Thus, we can leverage polymorphism to support such limited
but useful open reductions.

In the formal model, we add an \emph{opaque} construct to the target
language, written as
$\topaque\sexp$, which denotes an embedding of $\sexp$ at some
abstract type $\source{A}$ (as mentioned above, coercion to opaque is
a no-op in the implementation\guido{not mentioned}).
The coercion functions are extended
with rules to introduce and eliminate opaque terms (the $\sdelta$
superscript carries type substitutions as a technical device), e.g.:
\[
\begin{array}{c c c}
  \tcoerce{\svar{A}}{\sdelta}{\sexp}  = \topaque{\sexp}
  \quad & \quad
  \scoerce{\svar{A}}{\sdelta}{\topaque{\sexp}} = \sexp
  \quad
  &
  \quad
  \tcoerce{\sprod{\staui}{\stauii}}{\sdelta}{\spair{\sexpi}{\sexpii}} =
  \tpair{\prot{\source{e_1}}{\source{\tau_1}}{\sdelta}}{\prot{\source{e_2}}{\source{\tau_2}}{\sdelta}}
\end{array}
\]
where $\prot{\sexp}{\stau}{\sdelta}$ is $\topaque{\sexp}$ when
$\stau = \source{A}$, or a usual type-directed embedding otherwise. The source
semantics is 
extended with a rule for type applications, while target
semantics, being untyped, remains as-is.

Consider an application of the polymorphic identity function to a
variable $\svar{y}$:
$\sapp{\stapp{(\stabs{\svar{A}}{\sabs{\svar{x}}{\svar{A}}{\svar{x}}})}{\source{\mathbb{Z}}}}{\svar{y}}$.
As before, we would like to compile the function to native code and
then reduce the applications---except this time it is an open term.
Translation of the function and its embedding in the source yields
$\sapp{\stapp{\salien{\tabs{\tvar{x}}{\tvar{x}}}{\sall{\svar{A}}{\sarrow{\svar{A}}{\svar{A}}}}{\sesubstemp}}{\source{\mathbb{Z}}}}{\svar{y}}$.
We now apply the unembedding coercion
$\scoerce{\sall{\svar{A}}{\sarrow{\svar{A}}{\svar{A}}}}{\sesubstemp}{\tabs{\tvar{x}}{\tvar{x}}}$
which returns \linebreak a source type abstraction
$\stabs{\svar{A}}{\scoerce{\sarrow{\svar{A}}{\svar{A}}}{\sesubstemp}{\tabs{\tvar{x}}{\tvar{x}}}}$.
Reducing the arrow coercion next, we get
$\stabs{\svar{A}}{\sabs{\svar{x}}{\svar{A}}{\salien{\tapp{(\tabs{\tvar{x}}{\tvar{x}})}{\prot{\svar{x}}{\svar{A}}{\sesubstemp}}}{\svar{A}}{\sesubstemp}}}$,
where $\prot{\svar{x}}{\svar{A}}{\sesubstemp} = \topaque{\svar{x}}$, as mentioned above.
This is the key idea---when the original application reaches the beta-reduction with
argument $\svar{y}$ (after the type application is reduced), $\svar{y}$ is simply substituted
in this opaque construct, which is then returned back to the \fstar normalizer as
$\salien{\topaque{\svar{y}}}{\svar{A}}{\sesubstsingl{\svar{A}}{\source{\mathbb{Z}}}}$.
The normalizer can then apply the identity unembedding and get $\svar{y}$.

This ability to reduce open terms relying on polymorphism allows us to
evaluate expressions like
$\sapp{\sapp{\stapp{\stapp{\salien{\target{\text{\ls{List.map}}}}{\sall{\svar{A}}{\sall{\svar{B}}{\sarrow{\sarrow{(\sarrow{\svar{A}}{\svar{B}})}{\source{\text{\ls{list}}} \svar{A}}}{\source{\text{\ls{list}}} \svar{B}}}}}{\sesubstemp}}{\source{\mathbb{Z}}}}{\source{\mathbb{Z}}}}{\source{\text{\ls{f}}}}}{\source{\text{\ls{[0; ...; 10}$^6$\ls{]}}}}$,
with \ls$f$ a free variable,
using native datastructure implementations (\ls{List} in this case), 
which is much faster than using the normalizers.

\subsection*{Part 3: Full Details on our Multi-language Interoperability Model}

\begin{wrapfigure}[htb]{l}{\textwidth}%
      \centering
      \begin{framed}
        \input{interop_lazy/src_syntax.tex}
      \end{framed}
      \caption{Source Language}
      \label{fig:sourcelanglong}
\end{wrapfigure}

\begin{wrapfigure}[htb]{l}{\textwidth}%
      \centering
      \begin{framed}
        \input{interop_lazy/trg_syntax.tex}
      \end{framed}
      \caption{Target Language}
\end{wrapfigure}

\begin{wrapfigure}[htb]{l}{\textwidth}%
  \label{fig:srctyp}
  \fbox{$\iswf{\stenv}{\stau}$}
  \begin{framed}
    \input{interop_lazy/wf_typ}
  \end{framed}
  \caption{Type Well-formedness}
\end{wrapfigure}

\begin{wrapfigure}[htb]{l}{\textwidth}%
  \label{fig:srctyp}
  \fbox{$\iswf{\stenv}{\stau}$}
  \begin{framed}
    \input{interop_lazy/wf_subst}
  \end{framed}
  \caption{Type Substitution Well-formedness}
\end{wrapfigure}

\begin{wrapfigure}[htb]{l}{\textwidth}%
  \label{fig:srctyp}
  \fbox{$\hastyp{\stenv}{\senv}{\sexp}{\stau}$}
  \begin{framed}
    \input{interop_lazy/src_typing}
  \end{framed}
  \caption{Source Typing}
\end{wrapfigure}


\begin{wrapfigure}[htb]{l}{\textwidth}%
  \label{fig:srceval}
  \fbox{$\sexp \seval \sexpp$}
  \begin{framed}
    \input{interop_lazy/src_eval}
    where
    \input{interop_lazy/force}
  \end{framed}
  \caption{Evaluation in the Source Language}
\end{wrapfigure}

\begin{wrapfigure}[htb]{l}{\textwidth}%
  \label{fig:trgeval}
  \fbox{$\texp \teval \texpp$}
  \begin{framed}
    \input{interop_lazy/trg_eval}
  \end{framed}
  \caption{Evaluation in the Target Language}
\end{wrapfigure}

\medskip

\begin{wrapfigure}[htb]{l}{\textwidth}%
    \centering
    \begin{minipage}[t]{0.45\textwidth}\label{fig:stypsubst}
      \centering
      \begin{framed}
        \input{interop_lazy/stypsubst}
      \end{framed}
      \caption{Type Substitution in Source Terms}
    \end{minipage}\hfill
    \begin{minipage}[t]{0.45\textwidth}\label{fig:ttypsubst}
      \centering
      \begin{framed}
       \input{interop_lazy/ttypsubst}
      \end{framed}
      \caption{Type Substitution in Target Terms}
    \end{minipage}
\end{wrapfigure}

\medskip

\begin{wrapfigure}[htb]{L}{\textwidth}%
  \label{fig:trans}
  \fbox{$\sexp \trans \texp$} \\
  $\senv$ is a global parameter 
  \begin{framed}
    \input{interop_lazy/trans}
  \end{framed}
  \caption{Translation}
\end{wrapfigure}

\medskip

\begin{wrapfigure}[htb]{L}{\textwidth}%
    \centering
      \centering
      \begin{framed}
        \input{interop_lazy/scoerce}
      \end{framed}
      \caption{Target to Source Coercion}
      \medskip\medskip\medskip\medskip\medskip
      \centering
      \begin{framed}
        \input{interop_lazy/tcoerce}
      \end{framed}
      \caption{Source to Target Coercion}
\end{wrapfigure}

\medskip

\begin{wrapfigure}[htb]{l}{\textwidth}%
  \label{fig:protect}
  \begin{framed}
    \input{interop_lazy/protect}
  \end{framed}
  \caption{Protect meta-function}
\end{wrapfigure}

\medskip

\begin{wrapfigure}[htb]{l}{\textwidth}%
  \label{fig:typeq}
  \begin{framed}
    \input{interop_lazy/eqtyp}
  \end{framed}
  \caption{Type equivalence}
\end{wrapfigure}

\medskip

\begin{wrapfigure}[htb]{l}{\textwidth}%
  \label{fig:trans}
  \fbox{$\sexp \trans \texp$} \\
  \begin{framed}
    \input{interop_lazy/trans}
  \end{framed}
  \caption{Translation}
  \label{fig:translationlong}
\end{wrapfigure}

\fi

\bibliographystyle{abbrvnaturl}
\bibliography{fstar}

\end{document}